\documentclass[a4paper,twocolumn,floatfix,superscriptaddress,pra,longbibliography]{revtex4}
\usepackage[utf8]{inputenc}
\usepackage[T1]{fontenc}
\usepackage[sc,osf]{mathpazo}
\usepackage{amsmath, amsthm, amssymb,amsfonts,mathbbol,amstext}
\usepackage{graphicx}
\usepackage{dcolumn}
\usepackage{bm}
\usepackage{bbm}
\usepackage{hyperref}
\usepackage{mathtools}
\usepackage{comment}
\usepackage{color}
\usepackage{dsfont}
\usepackage[normalem]{ulem}
\usepackage{multirow}
\usepackage{diagbox}
\usepackage{amsmath}
\usepackage{amsthm}
\usepackage{bbm}
\usepackage{comment} 
\usepackage{marvosym}
\usepackage{graphicx}
\usepackage{xcolor}
\usepackage{appendix}
\usepackage{hyperref}
\usepackage{subcaption}
\usepackage{tikz} % TikZ
\usetikzlibrary{calc, arrows, patterns,shapes, decorations, arrows.meta, fit, positioning}
\tikzset{
    auto,node distance =1 cm and 1 cm,semithick,
    var/.style ={circle, draw, minimum width = 1cm, ultra thick},
    latent/.style ={regular polygon, regular polygon sides=3, inner sep=1pt, draw, minimum width = 1.2cm, ultra thick},
    point/.style = {circle, draw, inner sep=0.06cm, fill, node contents={}},
    triangle/.style = {regular polygon, regular polygon sides=3, draw, inner sep=0.06cm, fill, node contents={}},
    bidir/.style={Latex-Latex,dashed},
    dir/.style={-Latex, thick},
    el/.style = {inner sep=2pt, align=left, sloped}
}
\tikzstyle{vertex}=[circle, fill=black!10, draw=black]
\tikzstyle{edge}=[thick]
\tikzstyle{clique}=[line width=4, draw=black!70]

\newtheorem{theorem}{Theorem}

\newtheorem{lemma}[theorem]{Lemma}

\newtheorem{result}[theorem]{Result}

\definecolor{atomictangerine}{rgb}{1.0, 0.6, 0.4}
\definecolor{forestgreen}{rgb}{0.13, 0.55, 0.13}

\newcommand{\T}{\mathrm{T}}
\newcommand{\ket}[1]{\left| #1 \right\rangle}
\newcommand{\bra}[1]{\left\langle #1 \right|}

\newcommand{\dsep}[3]{#1 \perp_d #2 ~\vert ~ #3}
\newcommand{\pperp}{\perp \!\!\!\! \perp}
\newcommand{\cind}[3]{#1 \pperp #2 ~\vert ~ #3}
\newcommand{\stkout}[1]{\ifmmode\text{\sout{\ensuremath{#1}}}\else\sout{#1}\fi}

\def\minimize{\textrm{minimize}}
\def\st{\textrm{subject to }}

 \DeclareMathOperator{\tr}{tr}
\DeclareMathOperator{\Pa}{Pa}

\begin{document}
\title{Witnessing Non-Classicality in a Simple Causal Structure with Three Observable Variables}
\author{Pedro Lauand}
\email{p223457@dac.unicamp.br}
\affiliation{Instituto de Física “Gleb Wataghin”, Universidade Estadual de Campinas, 130830-859, Campinas, Brazil}
\author{Davide Poderini} 
\affiliation{International Institute of Physics, Federal University of Rio Grande do Norte, 59078-970, Natal, Brazil}
\author{Ranieri Nery} 
\affiliation{International Institute of Physics, Federal University of Rio Grande do Norte, 59078-970, Natal, Brazil}
\author{George Moreno} 
\affiliation{Departamento de Computação, Universidade Federal Rural de Pernambuco, 52171-900, Recife, Pernambuco, Brazil}
\author{Lucas Pollyceno}
\affiliation{Instituto de Física “Gleb Wataghin”, Universidade Estadual de Campinas, 130830-859, Campinas, Brazil}
\author{Rafael Rabelo}
\affiliation{Instituto de Física “Gleb Wataghin”, Universidade Estadual de Campinas, 130830-859, Campinas, Brazil}
\author{Rafael Chaves} 
\affiliation{International Institute of Physics, Federal University of Rio Grande do Norte, 59078-970, Natal, Brazil}
\affiliation{School of Science and Technology, Federal University of Rio Grande do Norte, Natal, Brazil}
\date{\today}

\date{\today}
\begin{abstract}
Seen from the modern lens of causal inference, Bell's theorem is nothing else than the proof that a specific classical causal model cannot explain quantum correlations. It is thus natural to move beyond Bell's paradigmatic scenario and consider different causal structures. For the specific case of three observable variables, it is known that there are three non-trivial causal networks. Two of those, are known to give rise to quantum non-classicality: the instrumental and the triangle scenarios. Here we analyze the third and remaining one, which we name the Evans scenario, akin to the causal structure underlying the entanglement-swapping experiment. We prove a number of results about this elusive scenario and introduce new and efficient computational tools for its analysis that also can be adapted to deal with more general causal structures. We do not solve its main open problem --whether quantum non-classical correlations can arise from it -- but give a significant step in this direction by proving that post-quantum correlations, analogous to the paradigmatic Popescu-Rohrlich box, do violate the constraints imposed by a classical description of Evans causal structure.
\end{abstract}

\maketitle
\section{Introduction}
Bell's theorem \cite{bell1964einstein} is a cornerstone of quantum theory, having far-reaching implications for its foundations as well as in applications for quantum information processing \cite{brunner2014bell}. The violation of a Bell inequality provides device-independent \cite{pironio2016focus} proof of the incompatibility of classical and quantum predictions, that is, solely based on the causal assumptions of an experiment and agnostic of any internal mechanisms of the involved physical apparatuses. More generally, it evidences the need for a genuine notion of quantum causal models \cite{fitzsimons2015quantum,ried2015quantum,fritz2016beyond,chaves2015information,costa2016quantum,allen2017quantum,aaberg2020semidefinite,wolfe2021quantum,ligthart2021convergent} in order to explain the correlations we observe in Nature.

Importantly, the mismatch between classical and quantum causal predictions can be generalized to causal structures beyond that in the paradigmatic Bell scenario. Motivated by the steady progress on quantum networks \cite{tavakoli2021bell}, there have been a number of results \cite{branciard2010characterizing,fritz2012beyond,tavakoli2014nonlocal,andreoli2017maximal,renou2019genuine,tavakoli2021bell,tavakoli2021bilocal,chaves2021causal,pozas2022full} proving that correlations across the distant parties of causal networks composed of independent sources can also exhibit non-classical behavior as already proven in a number of experiments \cite{carvacho2017experimental,saunders2017experimental,sun2019experimental,poderini2020experimental,carvacho2019perspective,cao2022experimental,suprano2022experimental}. In particular, quantum networks allow for a novel form of non-locality that, as opposed to Bell's theorem, does not require the need of measuring different observables \cite{fritz2012beyond,renou2019genuine,polino2022experimental}. In parallel, temporal scenarios based on causal structures involving communication between the parties \cite{brask2017bell,chaves2018quantum} have also provided a fruitful path for understanding the role of causality in quantum theory, for instance, showing that interventions --a central concept in the field of causal inference-- are able to reveal non-classicality in situations where Bell-like tests, based on observations, would simply fail \cite{gachechiladze2020quantifying,agresti2022experimental}.

It is thus natural to ask what causal structures can lead to non-classical behavior, a question that has remained elusive for two main reasons. The first is the fact that the number of possible causal structures increases very rapidly. Even to prove their equivalence classes -- that is, which causal structures can give rise to the same set of observed correlations-- has been solved up to three observable variables only \cite{evans2016graphs,ansanelli2022observational}. The other hurdle stems from the non-convex nature of the set of correlations permitted by general causal models \cite{garcia2005algebraic,geiger2013quantifier}. In spite of the number of complementary approaches developed in recent years \cite{chaves2014inferring,chaves2016polynomial,kela2019semidefinite,lee2017causal,WolfeSpekkensFritz+2019,pozas2019bounding,krivachy2020neural}, their practical use still is limited to a few cases of interest, which furthermore have to be evaluated on a case-to-case basis.

In the case of three observable nodes, it has been proved that there are a total of eight inequivalent classes of causal structures,
those depicted in Fig.~\ref{fig:nonlatent_dags} and Fig.~\ref{fig:all_dags}. From those, only three involve latent variables --that in a quantum description could be represented by entangled states-- and thus lead to correlations without a classical analog. Of these three, Fig.~\ref{fig:all_dags_instrumental} corresponds to the instrumental \cite{pearl1995testability,chaves2018quantum} and Fig.~\ref{fig:all_dags_triangle} to the triangle scenario \cite{fritz2012beyond,renou2019genuine}, bounded by Bell inequalities that can be violated with the help of entanglement, proving their non-classical nature. For the third causal structure, depicted in Fig.~\ref{fig:all_dags_evans} and to which we will refer as Evans causal structure \cite{evans2016graphs}, it is not known whether it can lead to non-classical correlations. 
We show that non-classical correlations reminiscent of PR-boxes \cite{popescu1994quantum} can violate the constraints implied by a classical description of this causal structure. The quantum violation of such bounds remains an open problem, however, as we show, a natural class of quantum correlations does have a classical explanation in such a scenario.
\begin{figure}
    \centering
    \begin{tikzpicture}
    \foreach \n in {2, 3} {
    \node[var] (c\n) at (2,\n*2.2 - 1) {$C$};
    \node[var] (b\n) at (1,\n*2.2) {$B$};
    \node[var] (a\n) at (0,\n*2.2 - 1) {$A$};
    }    
    \foreach \n in {4, 5} {
    \node[var] (c\n) at (6,\n*2.2 - 5.4) {$C$};
    \node[var] (b\n) at (5,\n*2.2 - 4.4) {$B$};
    \node[var] (a\n) at (4,\n*2.2 - 5.4) {$A$};
    }
    \node[var] (c1) at (4, 8) {$C$};
    \node[var] (b1) at (3, 9) {$B$};
    \node[var] (a1) at (2, 8) {$A$};
    % Directed edges
    \path[dir] (a1) edge (b1) (b1) edge (c1) (a1) edge (c1); 
    \path[dir] (a2) edge (b2) (c2) edge (b2);
    \path[dir] (a3) edge (b3) (b3) edge (c3);
    \path[dir] (b4) edge (c4); 
    \end{tikzpicture}
    \caption{\textbf{Possible tripartite causal structures with no latent variables}. Of the possible eight classes of causal structures with three observable node, five are represented by structures with no latent variables.
    }
    \label{fig:nonlatent_dags}
\end{figure}
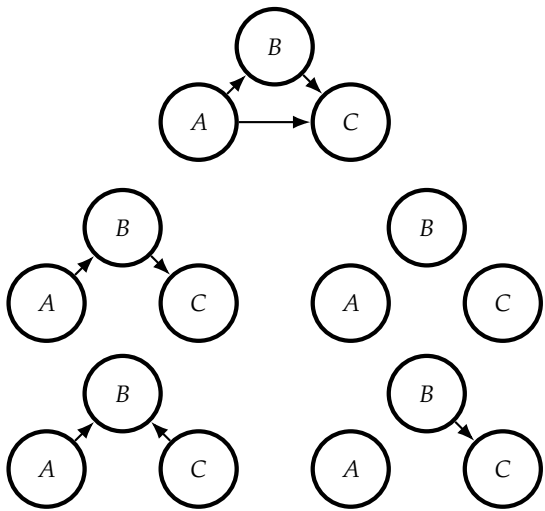

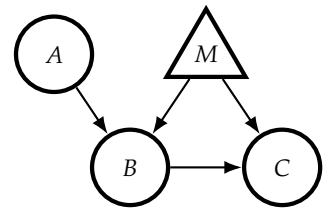
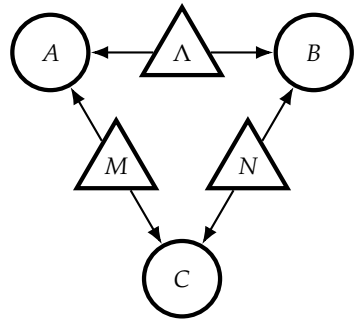
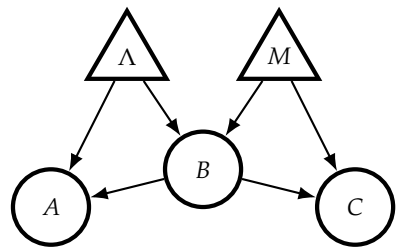
\begin{figure}
\begin{subfigure}[t]{.3\textwidth}
\begin{tikzpicture}
    \node[var] (b) at (2,0) {$C$};
    \node[var] (a) at (0,0) {$B$};
    \node[var] (x) at (-1,1.5) {$A$};
    \node[latent] (l) at (1,1.5) {$M$};
    % Directed edges
    \path[dir] (x) edge (a) (a) edge (b); 
    \path[dir] (l) edge (a) (l) edge (b);
\end{tikzpicture}
\caption{The Instrumental scenario}
\label{fig:all_dags_instrumental}
\end{subfigure}
\begin{subfigure}[t]{.3\textwidth}
    \begin{tikzpicture}

        \foreach [count=\k] \l/\n/\a in {1/A/\Lambda, 2/C/M, 3/B/N} {
            \draw (\k*360/3 - 30: 1cm) node[latent] (l\k) {$\a$};
            \draw (\k*360/3 + 30: 2cm) node[var] (\n) {$\n$};
			\path[dir] (l\k) edge (\n);
        }
		\foreach \k/\l in {1/B, 2/A, 3/C}
			\path[dir] (l\k) edge (\l);

    \end{tikzpicture}
    \caption{The Triangle scenario}
    \label{fig:all_dags_triangle}
\end{subfigure}
\begin{subfigure}[t]{.3\textwidth}
\begin{tikzpicture}
    \node[var] (a) at (-2,0) {$A$};
    \node[var] (c) at (2,0) {$C$};
    %\node[var] (b) [right =of a] {$B$};
    \node[latent] (l) at (-1,2) {$\Lambda$};
    \node[latent] (g) at (1,2) {$M$};
    \node[var] (b) at (0,.5) {$B$};
    % Directed edges
    \path[dir] (l) edge (a) (l) edge (b); 
    \path[dir] (g) edge (b) (g) edge (c); 
    \path[dir] (b) edge (a) (b) edge (c);
\end{tikzpicture}
    \centering        
    \caption{The Evans scenario}
    \label{fig:all_dags_evans}
\end{subfigure}
\caption{\textbf{Possible tripartite causal structures with latent variables}. There are a total of eight inequivalent classes of causal structures that can involve three observable nodes. Of these, only three of them contain at least one latent variable, a necessary condition to display a classical-quantum gap, that is, quantum correlations that cannot be explained by a classical causal model.}
\label{fig:all_dags}
\end{figure}

The paper is organized as follows. In Sec.~\ref{sec:dags} we introduce causal structures, their representation as directed graphs as well as a brief discussion of their equivalence classes. 
In Sec.~\ref{sec:evans} we prove some general results for the quantum and classical compatible distribution in the Evans scenario, using its similarity with the bilocality scenario. We further prove that, despite their similarities, a natural set of correlations in the bilocality scenario, involving measurements in the maximally entangled basis, do have a classical explanation in Evans, pointing out that the possible existence of genuine quantum distributions requires a more subtle approach.
In Sec.~\ref{sec:quad_prob} and \ref{sec: e-sep+inflation technique} we discuss two approaches to the causal compatibility problem: i) non-convex quadratic optimization and ii) the inflation technique \cite{wolfe2019inflation} augmented with e-separation \cite{Evans2012}, using both of them to demonstrate the nonclassicality of a  post-quantum distribution (similar to a PR-box) in the Evans scenario.
Finally, in Sec.~\ref{sec: Theory_independent_constraints} we discuss methods to derive conditions valid regardless of the nature of the latent sources (classical, quantum or post-quantum), which can effectively test the topology of the Evans causal structure.

\section{DAGs and their inequivalence classes}
\label{sec:dags}

A causal structure is represented by a directed acyclic graph (DAG) $\mathcal{G}$ which consists of a finite set of nodes $N_{\mathcal{G}}$ and a set of
directed edges $E_{\mathcal{G}}\subseteq N_{\mathcal{G}}\times N_{\mathcal{G}}$. These graphs need some distinction among the vertices to clarify if a node in the graph represents either an observable or an unobserved (or latent) variable. Graphically, we will use the circles to represent observable variables and triangles for latent variables see for example Fig.~\ref{fig:all_dags}), and we will denote the latter with greek letters.
We can define the concept of causal parents $\mathrm{Pa}(A)$ (or children $\mathrm{Ch}(A)$) of a given variable $A$ in $\mathcal{G}$, as the set of nodes sharing incoming (or outgoing) edges with $A$.
A joint distribution $p(\{a_i\}_i)$ over all the observable nodes $A_i \in O_{\mathcal{G}} \subset N_{\mathcal{G}}$ is said to be compatible with the DAG $\mathcal{G}$ if it satisfies the \emph{global Markov condition}, that is, it admits the following decomposition:
\begin{equation}
    p(\{a_i\}_i) = \sum_{\lambda \in L_{\mathcal{G}}} \prod_{v \in N_{\mathcal{G}}} p(v | \Pa(v))
    \label{eq:markov_condition}
\end{equation}
where $L_{\mathcal{G}} \subset N_{\mathcal{G}}$ is the set of latent variables.

An equivalent characterization of the compatibility of a distribution is given by the \emph{d-separation} criterion.
We first define the concept for an \emph{undirected path} in a DAG.
An undirected path $P \subset N_{\mathcal{G}}$ is a set of consecutive nodes of $\mathcal{G}$ connected by edges, regardless of their direction.
Such a path $P$ is said to be \emph{d-separated} by a set of nodes $X \subset N_{\mathcal{G}}$ if and only if
\begin{enumerate}
    \item There exists an $x \in X$ breaking at least a \emph{fork}
    ($a \leftarrow x \rightarrow b$) or a \emph{chain} ($a \rightarrow x
    \rightarrow b$) contained in $P$.
    \item If $P$ contains a \emph{collider} $a \rightarrow x \leftarrow b$ then
    neither $x$ nor its descendants are in $X$.
\end{enumerate}
Similarly, we say that two nodes $A, B$ (or equivalently, sets of nodes)
are \emph{d-separated} by a set of nodes $X$ if all paths connecting $A$
and $B$ are \emph{d-separated} by $X$, and it is denoted as $\dsep{A}{B}{X}$.
It can be shown~\cite{verma1988causal} that the concept of d-separation, when connected with conditional independence, completely characterizes the compatibility problem.
Specifically $p$ is compatible distribution for $\mathcal{G}$, iff for all disjoint subsets of vertices $X,Y,Z \subset N_{\mathcal{G}}$ we have $\dsep{X}{Y}{Z} \Leftrightarrow \cind{X}{Y}{Z}$, where $\pperp$ represent conditional independence.
We can denote the set of all compatible distributions with $\mathcal{G}$ as $\mathcal{C}_\mathcal{G}$.

Crucially, both the compatibility definitions described above are valid only if we assume that the latent variables behave like classical systems.
If we are free to identify the unobserved nodes as sources of quantum correlations, the set of quantum-compatible distributions is, for certain scenarios, strictly larger than the classical one.
Specifically, we can define quantum-compatible distribution for a DAG $\mathcal{G}$ as one described by the following strategy:
\begin{enumerate}
    \item To each latent variable $\Lambda \in L_{\mathcal{G}}$ we associate a quantum state described by the density operator $\rho_\Lambda \in \mathcal{L}(\mathcal{H}_\Lambda)$. 
    \item To each observed variable $A \in O_{\mathcal{G}}$ we associate a POVM measurement described by the operators $\left\{E_{\Pa^O(a)}^{a}\right\}_a$, dependent on the outcome of its observed parents $\Pa^O(a)$, and acting non-trivially only on the space $\mathcal{H}_{\Pa^L(A)}$ of all the latent parents of $A$. Measurement operators relative to different nodes should be pairwise commuting.
    \item The distribution is obtained by the Born rule applied to the state of all latent nodes :
    \begin{equation}
        p(O_{\mathcal{G}}) = \tr \left(\prod_{a \in O_{\mathcal{G}}} E_{\Pa^O(a)}^{a} \; \bigotimes_{\Lambda \in L_{\mathcal{G}}} \rho_\Lambda \right)
    \end{equation}
\end{enumerate}
Note that we are restricting to the case of exogenous latent variables.
We will denote the set of all quantum-compatible distributions for a DAG $\mathcal{G}$ as $\mathcal{Q}_\mathcal{G}$.  In this paper, we also consider resources that may not be compatible with a quantum description and are modeled by generalized probabilistic theories (GPT). A strategy similar to the one above can be used to define the set of correlations compatible with a causal structure $\mathcal{G}$ under scrutiny. Where we take $|\omega_{\Lambda})\in \Omega_{\Lambda}$, a GPT
generalization of the quantum state $\rho_{\Lambda}$, and $(e^{a}_{\Pa^{O}(a)}|$  GPT generalizations of the quantum measurement operator $E_{\Pa^O(a)}^{a}$ and the joint distribution over the observable variables is given by 
   \begin{equation}
        p(O_{\mathcal{G}}) =  \left(\prod_{a \in O_{\mathcal{G}}} (e_{\Pa^O(a)}^{a} \;| \circ  |\Omega)\right)
    \end{equation}
and $ |\Omega)$ lives in some composite state space $\prod_{\Lambda \in L_{\mathcal{G}}}\Omega_{\Lambda}$ that contains the tensor product of the states spaces as a subspace \cite{PhysRevA.81.062348}. We denote the set of all distributions compatible with a GPT as $\Omega_{\mathcal{G}}$.

Being interested in the \emph{causal compatibility problem} (CCP), it is convenient to restrict our attention to equivalence classes of DAGs that admit different sets of classically compatible distributions, that is we consider $\mathcal{G} \sim \mathcal{G'}$ belonging to the same class whenever $\mathcal{C}_\mathcal{G} = \mathcal{C}_\mathcal{G'}$.
Notice that while DAGs belonging to the same class need to have the same number of observable nodes, latent nodes are in no way restricted.
It can be shown \cite{Evans2016} that in the case of three observable nodes there are only eight different classes of DAGs, whose representatives are shown in Fig.~\ref{fig:nonlatent_dags} and Fig.~\ref{fig:all_dags}.
In particular in Fig.~\ref{fig:all_dags} are represented the three classes of DAGs containing latent variables, which are the only ones where a strict inclusion $\mathcal{C}_\mathcal{G} \subset \mathcal{Q}_\mathcal{G}$ is possible.
As previously mentioned, the Evan's scenario, depicted in Fig.~\ref{fig:all_dags_evans}, is the only one of them for this question remains an open problem.

\section{Evans Scenario}
\label{sec:evans}

% Tmp summary
% Evans scenario (markov fact for later ref)
% First - connection with bilocality get several results
% Bilocality inspire ent. swapping strategy -> we prove it cannot work
% Second - connection with the instrumental
% Instrumental inspires NS disitribution, incompatibility found with gurobi
% To find a witness we use  inflation with e-sep
%\textcolor{red}{[We should have a one paragraph here explaining what we will do in each of the subsections]}
\subsection{The relation between the bilocality and the Evans scenario}

In the following, we show an interesting relation between the Evans scenario to the so-called bilocal scenario~\cite{PhysRevLett.104.170401}.
This scenario, whose DAG $\mathcal{G}_B$ is displayed in Fig.~\ref{fig:bilocality_dag}, presents two independent sources $\Lambda$ and $M$ shared between a central node $B$ and two peripheral nodes $A$ and $C$, and can be described by a conditional probability distribution $p(a,b,c \vert x,z)$.
It is known~\cite{PhysRevLett.104.170401} that for this scenario the classically compatible set of distributions does not coincide with the quantum set, i.e. $\mathcal{C}_B \subset \mathcal{Q}_B$.

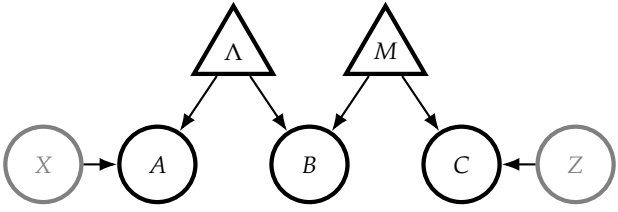
\begin{figure}
\centering
    \begin{tikzpicture}
        \node[latent] (l1) at (-1,1.5) {$\Lambda$};
        \node[latent] (l2) at (1,1.5) {$M$};
        \node[var] (a1) at (-2,0) {$A$};
        \node[var] (b) at (0,0) {$B$};
        \node[var] (a2) at (2,0) {$C$};
        \node[var, opacity=.5] (x1) at (-3.5, 0) {$X$};
        \node[var, opacity=.5] (x2) at (3.5, 0){$Z$};
        
        \path[dir] (x1) edge (a1) (x2) edge (a2);
        \path[dir] (l1) edge (a1) (l1) edge (b) (l2) edge (b) (l2) edge (a2);
    \end{tikzpicture}
     \caption{\textbf{ Causal structure of the bilocality scenario.} Two independent
     sources $\Lambda$ and $M$ connect the nodes $A$, $C$, which are influenced by $X$, $Z$ respectively, to $B$.}
\label{fig:bilocality_dag}
\end{figure}

A distribution compatible with the Evans scenario can be seen as a projection of a distribution in the bilocal scenario into the subspace $p(a,b,c|x=z=b)$. 
More precisely, for every distribution admitting a classical model $p_E(a,b,c)$ in the Evans scenario there exists a realization $p_B(a,b,c|x,z)$ in the bilocal scenario such that $p_E(a,b,c) = p_B(a,b,c|x=z=b)$ and $p_B(a,b,c|x,z)$ also admits a classical model. 
Conversely, for every bilocal distribution $p_B(a,b,c|x,z)$ its projection, $p_E(a,b,c) := p_B(a,b,c|x=z=b)$, admits classical model in Evans.

This simple fact allows us to derive some non-trivial conclusions about the topology of the set of classical correlations in Evans, $\mathcal{C}_e$, that follow from the topology of the bilocal set. We show, in analogy to \cite{PhysRevA.85.032119}, that the set $\mathcal{C}_e$ is connected and has weak star-convexity for a certain subspace. A similar relation is respected under the same projection for the case where quantum states are distributed in the network. The combination of these results shows that quantum non-bilocality is a necessary condition for a given distribution to display a classical-quantum gap in the Evans scenario.

To formalize those statements, let us consider the Evans set $\mathcal{C}_e$ and the bilocal set $\mathcal{C}_{B}$ of classical correlations. 
We know that for a distribution $p(abc|xz)\in \mathcal{C}_{B}$ the \emph{global Markov property} implies:
\begin{equation}
    p(abc|xz)=\sum_{\mu,\lambda}p(\mu)p(\lambda)p(a|x\lambda)p(b|\mu\lambda)p(c|z\mu) \; .
    \label{eq:evans_factorization}
\end{equation}

From this, assuming that $|X|=|Z|=|B|$, it follows
\begin{lemma}
\label{lemma:1}
If $p_B(abc|xz)\in\mathcal{C}_{B}$ then $p_E(abc):=p(abc|x=z=b)\in \mathcal{C}_e$. 
Conversely, if $p_E(abc)\in\mathcal{C}_e$ then it exists a $p_{B}(abc|xz)\in \mathcal{C}_{B}$ such that $p_{B}(abc|x=z=b)=p_E(abc)$.
\end{lemma}
\begin{proof}
The first implication of the lemma follows directly from the definitions. Suppose now that $p_E\in \mathcal{C}_e$ then there exists $\lambda$ and $\mu$ such that  
\begin{equation}
p(abc)=\sum_{\mu,\lambda}p(\mu)p(\lambda)p(a|b\lambda)p(b|\mu\lambda)p(c|b\mu).    
\end{equation}
Now we can define
\begin{equation}
\tilde{p}(a|\lambda,x):=p(a|\lambda,b=x),
\end{equation}
and similarly for $\tilde{p}(c|\mu,z)$.  We can also define
\begin{multline}
\tilde{p}(a,b,c|x,z) = \\ = \sum_{\mu,\lambda}p(\mu)p(\lambda)p(a|b'=x,\lambda)p(b|\mu\lambda)p(c|b'=z,\mu),  
\end{multline}
where $\tilde{p}$ is bilocal by construction and $\tilde{p}(a,b,c|x=z=b)=p(abc)$, which concludes the proof.
\end{proof}

Using this mapping we can extend some results valid for the set of bilocal distributions, to the Evans one.
In particular, we can prove that $\mathcal{C}_e$ is connected and star-convex for certain subspaces, i.e. there exists a preferential point $p^{*} \in \mathcal{C}^{p_A}_e \subset  \mathcal{C}_e$ such that for any
$p \in \mathcal{C}^{p_A}_e$ the line from $p^{*}$ to $p$ is contained in $\mathcal{C}^{p_A}_e$. Notice that star-convexity is a weaker notion of convexity.
\begin{lemma}
\label{lemma:2}
$\mathcal{C}_e$ is connected 
\end{lemma}
\begin{proof}
First, it was shown in \cite{PhysRevA.85.032119} that $\mathcal{C}_B$ is connected. More precisely, for every bilocal correlation $p_B$ there is a correlation $p_{\xi}=\xi p_B+(1-\xi)p_0$ that follows a continuous path connecting $p_B$ to $p_0$, where  $p_0(a,b,c|x,z)=\dfrac{1}{|A||B||C|}$ is the uniform distribution. Furthermore, $p_{\xi}$ is bilocal for every $\xi$ as it can be obtained by performing local operations on the correlation $p_B$. 

Now, consider $p_E\in\mathcal{C}_e$. Then by lemma 1 there exists $p_B$ bilocal that recovers $p_E$ via projection, we can thus define $$p^E_{\xi}(a,b,c)=p_{\xi}(a,b,c|x=z=b)=\xi p_E+(1-\xi)\dfrac{1}{|A||B||C|},$$since $p_{\xi}$ is bilocal then $p^E_{\xi}\in \mathcal{C}_e$. This shows every classical $p_E$ is continuously connected to the uniform distribution, implying that that $C_e$ is connected.
\end{proof}

Moreover, we have that star-convexity is preserved when we project on the subspace determined by fixing a specific marginal distribution for a $p_A(a)$.
\begin{lemma}
\label{lemma:3}
$\mathcal{C}_e$ is star-convex on the subspace of fixed $p_A$ denoted by $\mathcal{C}^{p_A}_e$.
\end{lemma}
\begin{proof}
We need to show there exists a point $p^{*}\in \mathcal{C}^{p_A}_e$ such that any line segment between $p\in \mathcal{C}^{p_A}_e $ and $p^{*}$ is inside $\mathcal{C}^{p_A}_e$.

Consider $p\in \mathcal{C}^{p_A}_e $, i.e.,  $p(a)=\sum_{bc}p(a,b,c)=p_A(a)$ is fixed and $p$ has a classical decomposition in terms of $\lambda$ and $\mu$. Define $p^{*}(a,b,c)=p_A(a)p^{*}(b)p^{*}(c)\in  \mathcal{C}^{p_A}_e $.

By providing an extra random bit $\ell$ distributed by the source with probability $p(\ell=1)=1-p(\ell=0)=\xi\in[0,1]$ to Bob and Charlie we can define the correlation $p_{\xi}$ as follows: if $\ell=1$, Bob and Charlie output their original response functions defined by $p$. If $\ell=0$ they respond according to $p^*$. This yields $p_{\xi}=\xi p +(1-\xi)p^* $. This operation can be done locally from Bob and Charlie labs and thus we can construct a local model for $p_{\xi}$, and $\sum_{bc}p_{\xi}(a,b,c)=p_{\xi}(a)=p_A(a)$. This shows $p_{\xi}\in \mathcal{C}^{p_A}_e $. 
\end{proof}

The proof of lemma 3 uses the exact same argument as Appendix A in \cite{PhysRevA.85.032119}. Lemmas 2 and 3 imply similarities between  the sets $\mathcal{C}_e$ and $\mathcal{C}_B$. Our next result shows there is a similar relationship at the quantum level.
The set of quantum correlations in the Evans scenario, $\mathcal{Q}_e$, and in the bilocal, $\mathcal{Q}_B$, are given respectively by 
\begin{equation*}
p_E(a,b,c)=\tr(\psi_{AB}\otimes\psi_{B'C}(E_{a|b}\otimes E_b \otimes E_{c|b}))\in \mathcal{Q}_e    
\end{equation*}
and 
\begin{equation*}
    p_B(a,b,c|x,z)=\tr(\psi_{AB}\otimes\psi_{B'C}(E_{a|x}\otimes E_b \otimes E_{c|z}))\in \mathcal{Q}_B
\end{equation*}
The systems are independently  distributed and $\{E_{i|j}\}$ are general POVM's respecting 
$$\sum_i E_{i|j}=1 \quad \forall j \quad E_{i|j}\geq0 \quad \forall i,j.  $$

\begin{lemma}
If $p_B(abc|xz)\in\mathcal{Q}_{B}\implies p_E(abc):=p(abc|x=z=b)\in \mathcal{Q}_e$. Conversely,
if $p_E(abc)\in\mathcal{Q}_e\implies \exists p_{B}(abc|xz)\in \mathcal{Q}_{B}$ s.t. $p_{B}(abc|x=z=b)=p_E(abc)$
\end{lemma}
\begin{proof}
This is a consequence of making the identification $E_{a|x=k}\equiv E_{a|b=k}$, and similar for other effects. 

Indeed, if $p_B(abc|xz)\in\mathcal{Q}_{B}$ then, 
\begin{equation*}
\begin{aligned}
p_E(abc) &:=p(abc|x=z=b)=\\
&=\tr(\psi_{AB}\otimes\psi_{B'C}(E_{a|b}\otimes E_b \otimes E_{c|b}))\in \mathcal{Q}_e.
\end{aligned}
\end{equation*}
Conversely, if $p_E\in \mathcal{Q}_e$ we define $\tilde{E}_{a|x}:= E_{a|b=x}$, and similar for other parties, which gives 
\begin{equation*}
    \tilde{p}(a,b,c|x,z)=\tr(\psi_{AB}\otimes\psi_{B'C}(E_{a|b'=x}\otimes E_b \otimes E_{c|b'=z}))\in \mathcal{Q}_B
\end{equation*}
\end{proof}
It is important to point out that we have proven quantum non-bilocality to be a necessary condition for the Evans scenario quantum gap. 

Indeed, suppose $p_E$ is a quantum non-classical distribution in the Evans scenario. By lemma 4 there exists $p_B\in \mathcal{Q}_B$ that recovers $p_E$ via projection. The distribution $p_B$ is not unique, i.e. the quantified problem 
\begin{equation}
    \exists p_B \in \mathcal{Q_B} \text{ s.t. } p_B(a,b,c|x=z=b)=p_E(a,b,c),
\end{equation}
may have many (infinitely many) solutions. However, there cannot be any solution $p_B\in\mathcal{C}_B$, because if $\exists p_B\in\mathcal{C}_B$ then $p_E\in\mathcal{C}_e$ by lemma 1. A contradiction.

These relations hint at the possibility that starting with some non-classical distribution in the bilocal scenario one could possibly derive the incompatibility of the projected correlation in the Evans case, i.e., it is natural to ask whether non-bilocality could become also sufficient under some specific condition or some specific distribution. 
Unfortunately, we could not find quantum violations of bilocality that remained non-classical after projection. 
We believe that this is due to the fact that many examples of non-bilocality rely on entanglement swapping processes that are heavily dependent on maximally entangled measurements for Bob, which, as we will prove in the next section, do have a classical explanation in the Evans case.

As shown by these negative results, truly new methods are required to analyze the emergence of non-classicality in the Evans scenario.

\subsection{A classical model for measurements on a maximally entangled basis}
\label{sec: classical_model-entangled_meas}
Given the similarities between the bilocal and the Evans scenario that we highlighted in the previous section, a natural quantum strategy to obtain nonclassical correlations in the Evans scenario would be to start from maximally entangled bipartite states shared by $A,B$ and $B,C$ and use node $B$ to perform entanglement-swapping~\cite{zukowski1993event}, so that $A$ and $C$ would effectively perform measurements on a shared entangled state.
Unfortunately, we will show here that this kind of strategy possesses a classical model in the Evans scenario.

To make things more precise let us assume that $A,B$ and $B,C$ both start with maximally entangled states of the form $\ket{\Phi_d} = \frac{1}{\sqrt{d}}\sum_i^d \ket{ii}$, where $B$ performs the standard $d$-dimensional entangling measurement, in the basis $B_d = \left\{\ket{\Phi_d^{n,m}}\right\}_{m,n}$
\begin{equation}
    \ket{\Phi_d^{n,m}} = \frac{1}{\sqrt{d}}\sum_k^d e^{i 2\pi nk/d} \ket{k,k + m}
    \label{eq:ent_measurement}
\end{equation}
where the sum is modulo $d$ and $n,m \in \{0, \ldots, d-1\}$.
Then $A$ and $C$ can perform an arbitrary POVM measurement on their part $\{A_b^a\}_a$ and $\{C_b^c\}_c$ respectively, depending on the outcome $b = (n,m)$ of $B$.
The distribution is given by
\begin{equation}
    p(a,b,c) = \tr \left( A_b^a \otimes B^b \otimes C_b^c \,
        \ket{\Phi_d}\bra{\Phi_d} \otimes \ket{\Phi_d}\bra{\Phi_d} \right )
    \label{eq:max_ent_st}
\end{equation}

\begin{result}
For any POVM $\{A_b^a\}_a$ and $\{C_b^c\}_c$ and for any dimension $d$, the distribution generated by the quantum strategy in equation \eqref{eq:max_ent_st} always has a classical realization compatible with the Evans causal structure.
\end{result}
\begin{proof}
We can start by noticing that independently of the value of the outcome $b$, we can always describe the effective state shared between $A$ and $C$ as the maximally entangled state $\ket{\Phi_d}$.
Indeed, since depending on $b$, the effectively swapped state will be one of the set $\{\ket{\Phi_d^{n,m}}\}$, we can reduce it to $\ket{\Phi_d}$ by applying a local unitary on one part. Specifically, we can repeatedly apply the operators $X, Z$, defined by
\begin{equation}
    X \ket{k} = \ket{k + 1}, \quad Z \ket{k} = e^{i2\pi k/d} \ket{k} \; .
\end{equation}
In this way, redefining 
\begin{equation}
    \tilde C_{(n,m)}^c = Z^{-n} X^{-m} C_{(n,m)}^c X^m Z^n
\end{equation}
we obtain, conditioning on the outcome $b$
\begin{equation}
        p(a,c|b) = \tr \left( A_b^a \otimes \tilde C_b^c \, \ket{\Phi_d}\bra{\Phi_d} \right)
        = \tr \left( (A_b^a)^\T \tilde C_b^c \right)/d \; .
        \label{eq:proof_max_ent_pacb}
\end{equation}

The distribution $p(a,c|b)$, as written above, is also compatible with a structure where variables $A$ and $C$ share a single latent variable $N$, and are both influenced by the same setting $B$ (independent of $N$).
In such a structure the sets of quantum and classical distribution sets coincide, so we know that in particular for \eqref{eq:proof_max_ent_pacb}, we can have a classical realization of the form $p(a,c|b) = \sum_\nu p(a|b,\nu) p(c|b,\nu) p(\nu)$, where $\nu \in \{0,\ldots,d-1\}^2$.
Without loss of generality, we can limit ourselves to combinations of deterministic strategies:
\begin{equation}
    p(a,c|b) = \sum_{\nu} \delta_{a, f(b,\nu)} \delta_{c, g(b,\nu)} p(\nu) \, .
    \label{eq:proof_max_ent_pacb_nu}
\end{equation}
We will now show that these can be easily converted into classical realizations compatible with the Evans scenario.
Indeed for a classical distribution in the Evans case, we would have 
\begin{equation}
    p(a,c|b) = \sum_{\lambda, \mu} \delta_{a, f'(b,\lambda)} \delta_{c, g'(b,\mu)} p(\lambda, \mu | b),
    \label{eq:proof_max_ent_pacb_evans}
\end{equation}
where the term $p(\lambda, \mu | b)$ needs to satisfy  $\sum_b p(\lambda, \mu | b) p(b) = p(\lambda) p(\mu)$.
We notice that in our case we have $p(b)$ a uniform distribution, and if we also choose $p(\lambda)$ and $p(\mu)$ to be uniform we have that $p(\lambda, \mu | b) = p(b | \lambda, \mu)$
so that the above condition becomes simply the normalization of $p(b | \lambda, \mu)$.
Now, if $p(\nu)$ is the distribution for the shared latent variable $N$ that describes \eqref{eq:proof_max_ent_pacb_nu} 
, we define
\begin{equation}
    p(b | \lambda, \mu) = p(\nu = (\lambda + b_1, \mu + b_2)),
\end{equation}
where $b = (b_1, b_2)$ and the addition is modulo $d$.
We can easily see that, thanks to this definition, both $p(b | \lambda, \mu)$ and $p(\lambda, \mu | b)$ satisfy the normalization condition.
Now, to make \eqref{eq:proof_max_ent_pacb_evans} coincide with \eqref{eq:proof_max_ent_pacb_nu}, only remains to define the local deterministic strategy of the former as $f'(\lambda, b) = f(\lambda - b_1, b)$ and $g'(\mu, b) = g(\mu - b_2, b)$.
\end{proof}
It's important to notice that our result is valid regardless of the cardinality of the variables $A$ and $C$. It was shown in \cite{bonet2013instrumentality} that, in the instrumental scenario, when the central node (in our case $B$) has a high cardinality the scenario becomes less restrictive such that in the asymptotic limit of continuous $B$ the model imposes no restriction for $p(a,b,c)$. Since the Evans model contains the instrumental model as a particular case (as will be discussed in more detail in \ref{sec:InstrumentalIneqforEvans}), we can conclude that this will also be the case for the Evans model. This means that whenever $|B|\geq |A|,|C|$ the model will be quite permissive and we can, for example, simulate perfect correlation between the variables trivially.
\section{Non-convex quadratic problems and causal compatibility}
\label{sec:quad_prob}

In our framework, non-convex optimization problems arise in a very natural way: to solve the problem of deciding whether some correlation $p$ lies inside or outside the non-convex set of classical correlations. It has been proven in~\cite{Rosset2018UniversalBO} that the classical set of correlations, for any network structure, is a semi-algebraic set, i.e. it can be described by finitely many polynomial equalities and/or inequalities on the joint distribution of all observable nodes of the network. Thus, in principle, we could solve the CCP by globally solving a \emph{polynomial optimization problem} (POP). Of course, in general, this is no easy task and there are approaches to try to solve such problems asymptotically as a hierarchy of simpler outer approximations like, for example, semi-definite programming relaxations~\cite{Lasserre2001GlobalOW} or  
Linear Programming (LP)~\cite{NavascuesWolfe+2020+70+91}.

Reformulating the Evans compatibility problem 
as a non-convex maximization problem, we get that $p(a,b,c)$ is compatible with \ref{fig:all_dags_evans} if and only if 
\begin{equation}\label{eq: Fine theorem}
\begin{aligned}
    \exists &q(a_0,...,a_{|B|}, b, c_0,...,c_{|B|}) \geq 0\\
    &\st\\
    &\sum_{a_{i\neq b},c_{j\neq b} }q(a_0,...,a_{|B|}, b, c_0,...,c_{|B|})=p(a,b,c),\\
    & q(a_0,...,a_{|B|}, c_0,...,c_{|B|})=q(a_0,...,a_{|B|})q( c_0,...,c_{|B|}),\\
    &\sum_{a_{0},...,a_{|B|},b,c_{0},...,c_{|B|}}q(a_0,...,a_{|B|}, b, c_0,...,c_{|B|})=1,\\
   \end{aligned}
\end{equation}
where 
\begin{equation}
\begin{aligned}
   &q(a_0,...,a_{|B|}, b, c_0,...,c_{|B|}):=\sum_{\lambda,\mu}p(\lambda)p(\mu)p(b|\lambda,\mu)\times \\
   &\times p(a_0|\lambda,b=0)...p(a_{|B|}|\lambda ,b=|B|) p(c_0|\mu,b=0)...\\
   &...p(c_{|B|}|\mu ,b=|B|).
\end{aligned}
\end{equation}
This problem is non-convex, with all the conditions being at most quadratic. We also can see in~\ref{eq: Fine theorem} analogous conditions to Fine's theorem~\cite{fine1982hidden} for this network, i.e. necessary and sufficient conditions on the joint probability distribution of all the possible outcomes for causal compatibility. This clearly constitutes a non-convex quadratically constrained feasibility test that can be globally solved with presently available solvers.
Notice that, feasibility tests are the particular case of a constrained optimization problem, that is, when the objective function is constant. 
For our work, we have used the Gurobi optimizer~\cite{gurobi}, which allows us to tackle non-convex quadratically constrained quadratic problems (QCQP).
Contrary to many non-linear optimization solvers, which search for locally optimal solutions, here we solve this problem by looking for global optimality

\footnote{Gurobi supports constraints containing bilinear terms like $x\cdot y$. Although it does not directly support constraints containing more general multilinear terms, they can be modeled using a series of bilinear constraints, with the help of auxiliary variables. Such higher-degree polynomial problems will face numerical issues due to computational precision. }.

Furthermore, we can rewrite this as a minimization problem introducing white noise to $p(a,b,c)$:

\begin{equation}\label{op:OptimizationGurobi}
\begin{aligned}
    &\minimize_{v,q} \quad  1-v\\
    &\st\\
    &\sum_{a_{i\neq b},c_{j\neq b} }q(a_0,...,a_{|B|}, b, c_0,...,c_{|B|})=vp(a,b,c)+\frac{1-v}{|A||B||C|},\\
    & \quad \quad q(a_0,...,a_{|B|}, c_0,...,c_{|B|})=q(a_0,...,a_{|B|})q( c_0,...,c_{|B|}),\\
    &\sum_{a_{0},...,a_{|B|},b,c_{0},...,c_{|B|}}q(a_0,...,a_{|B|}, b, c_0,...,c_{|B|})=1,\\
    & \quad \quad q\geq 0.
   \end{aligned}
\end{equation}

Notice that this choice of noise is validated by lemma~\ref{lemma:2}, where it is shown that all Evans-compatible behaviors are continuously connected to the white noise distribution and, thus, the problem must yield a solution with $v<1$ if $p$ lies outside the classical set of correlations.  
Although the derivation of~\ref{eq: Fine theorem} does not work in general for networks with a more complex latent structure, e.g. triangle network, our tools may still tackle the CCP for these more general causal structures.
Truly, in~\cite{Rosset2018UniversalBO} the authors show that one may always take the cardinality of the latent variables to be finite and, by consequence, we can use the d-separation criterion to formulate a POP. This shows that we can use global non-convex optimization as a necessary and sufficient test, up to computational tolerance, for causal compatibility.

\subsection{Extracting infeasibility certificates from quadratic problems}
\label{sec:Extractinginfeasibility}
If a given probability distribution $p^*(a,b,c)$ fails the CCP test, i.e. is incompatible with the classical causal model under test, we would also like to extract a witness for such non-classical behavior, that is, a real function $F$ such that $F(p(a,b,c))\geq 0$ and $F(p^*)<0$ where $p(a,b,c)$ are all distributions compatible with the causal structure under scrutiny. 

In the Evans scenario, for example, the compatibility conditions are stated in~\ref{eq: Fine theorem}. Since any compatible $p(a,b,c)$ can be expressed as a marginalization of the joint distribution $q$, we can write $F=F(q)$, where the terms that will appear are the respective marginals of $q$ optimized over the factorizing conditions, normalization and non-negativity. 
Given an incompatible distribution $p^*$, we may simply choose $F(p)=||p-p^*||^2$, since all terms of $F$ are quadratic on $q$ and this can be efficiently optimized. 
Notice that the value $F(p)=0$ is never possible, as this would imply $p=p_*$, and $p^*$ is assumed to be a non-feasible point. Therefore, the program will return a bound $B^*$ that is tight up to computational precision. Our witness will then be written as $$F(p)\geq B^*$$ and can be violated by $p_*$.

\subsection{Finding inequalities for Evans starting from the instrumental scenario}

\label{sec:InstrumentalIneqforEvans}

The instrumental scenario, whose DAG is shown in Fig.~\ref{fig:all_dags_instrumental}, implies the following decomposition for any compatible distribution
\begin{equation}
p(a,b,c) = \sum_{\mu}p(a) p(b\vert a,\mu)p(c\vert b,\mu)p(\mu). 
\end{equation}
We notice that this can be regarded as a particular case of the factorization~\eqref{eq:evans_factorization} given by Evans scenario.

The correspondence emerges when we drop the arrow from $B$ to $A$ and choose $A$ to be a deterministic function of $\lambda$, that is,  $p(a \vert b, \lambda)= \delta_{a,\lambda}$. 
More generally, the arrow $B \rightarrow A$ can be seen as a kind of measurement dependence in the instrumental scenario \cite{miklin2022causal}, since via this arrow the hidden variable $\mu$ can influence the variable $A$. Because of that, it seems reasonable that inequalities valid for the instrumental scenario can be recycled for Evans if we consider the relaxation arising from the arrow $B \rightarrow A$. 
In the following, we will show how the bounds of the instrumental inequalities change if applied to the Evans scenario.

Indeed, using the Gurobi optimizer, we can evaluate how the bound of known Bell inequalities change when applied to Evans. 
Specifically, we consider Pearl's inequality \cite{pearl1995testability}
\begin{align}\label{eq:binary_instrumental_ineqs2}
    P = &p_{ABC}(1,0,0)p_A(1)+p_{ABC}(1,1,1)p_A(0)\\
    &- p_A(0)p_A(1) \leq 0, \nonumber
\end{align}
and Bonet's inequality \cite{bonet2013instrumentality, chaves2018quantum}
\begin{align}
\label{eq:Bonet2} 
    \nonumber B = &p_{ABC}(0,1,0)p_A(1)p_A(2)-p_{ABC}(1,1,0)p_A(0)p_A(2)\\
    \nonumber &-p_{ABC}(1,1,1)p_A(0)p_A(2)-p_{ABC}(2,0,1)p_A(0)p_A(1)\\  
    &-p_{ABC}(2,1,0)p_A(0)p_A(1) \leq 0.
\end{align}

Employing Gurobi, one can see that the inequalities are changed as
\begin{eqnarray}
    &P \leq 1/16, \\ \nonumber
    &B \leq 1/27.
\end{eqnarray}
Unfortunately, however, we were unable to find any quantum or post-quantum correlations able to violate these inequalities.
.

Nonetheless, this relation between Evans and the instrumental also hints at the possibility that correlations violating instrumental inequalities might also lead to non-classical behavior in the Evans scenario.  As will be detailed in the following, this is precisely the case for the PR-box in the instrumental scenario with $|A|=3$, appropriately choosing the marginal distribution of $p(a)$. We also show in Appendix \ref{app: Classical_PR-Box} how this choice cannot be arbitrary since there are choices for $p(a)$ such that the distribution still admits a local model.

\subsection{Violations of the Evans' scenario} 

Now we show that the Bonet PR-Box given by  
\begin{equation}
    p_{PR}(bc|a)=
    \begin{cases}
    1/2 \quad \textrm{ if $c=b+f(a,b)$} \mod{2}\\
    0\quad  \textrm{ otherwise}
    \end{cases},
\label{eq:NS_Box}
\end{equation}
where $f(0, b) = 0$, $f(1, b) = b$, and $f(2, b) = b +1 \mod{2}$, can violate causal compatibility inequalities in the corresponding Evans scenario.

Solving the QCQP, we have a necessary and sufficient (up to computational precision) oracle to detect non-classically in the Evans scenario.
By discretizing our domain of parameters $p_A$, and employing this oracle for each choice of $p_A$, the lowest visibility overall choices we could certify the non-classicality of the distribution in~\eqref{eq:NS_Box} is given by $v=0.84$. With $p_A(0)=p_A(2)=10/21$ and  $p_A(1)=1/21$, where the distribution under test is the convex sum $p(a,b,c)=vp_{PR}(a,b,c)+(1-v)1/12$, that is, a mixture of white noise with the PR-box wired by asymmetrically distributed inputs. 

Using the construction in Sec. \ref{sec:Extractinginfeasibility} we could, in theory, always find a non-linear witness of non-classicality which is generally valid for the Evans scenario. However, in practice, this approach may face numerical difficulties in regions near critical visibilities. This fact becomes more intuitive when we look at this approach in a geometrical manner, since we are looking at the largest euclidean ball centered at the nonclassical correlation $p$ that remains completely outside of the non-convex (classical) set of correlations. Clearly, this is a non-convex optimization problem and, furthermore, if we get arbitrarily close to the border of the set from outside, the radius of this ball becomes arbitrarily small, and if this radius is smaller than our computational precision we have a problem. Due to these difficulties, we chose to look to the largest radius of this euclidean ball, which would correspond to the case $v=1$. After some simplification of the terms, we find
\begin{equation}
\begin{aligned}
&GW=p^2_{ABC}(000) - \dfrac{10}{21}p_{ABC}(000) + p^2_{ABC}(001)+\\
& p^2_{ABC}(010) + p^2_{ABC}(011) - \dfrac{10}{21}p_{ABC}(011) + p^2_{ABC}(100)+ \\
&p^2_{ABC}(101)- \dfrac{1}{21}p_{ABC}(100)  + p^2_{ABC}(110) - \dfrac{1}{21}p_{ABC}(110) \\
&+ p^2_{ABC}(111) + p^2_{ABC}(200) + p^2_{ABC}(201) - \dfrac{10}{21}p_{ABC}(201)+\\
& p^2_{ABC}(210) + p^2_{ABC}(211) - \dfrac{10}{21}p_{ABC}(211) +0.22728 \geq 0. 
\end{aligned}
\end{equation}
To obtain this witness, we optimized the objective function $F(p)=||p-p_{PR}||_2^2$ where $p_{PR}$ is the PR-box distribution mentioned before with the appropriate choice of $p_A$, the optimal violation is the case where $p=p_{PR}$ which we get $GW(p_{PR})=-0.00061$. We can also always show with our approach that this inequality is tight, i.e., gurobi program can always return a classical point that achieves the optimal objective function. More generally, this can be seen as a consequence of the Weierstrass extreme value theorem which guarantees that the continuous function $F(p): \mathcal{C}_e\rightarrow \mathbb{R}$ can take its extreme values with points inside the compact set of classical correlations $\mathcal{C}_e$.

\section{Deriving a witness from the Inflation technique augmented with E-separation}

\label{sec: e-sep+inflation technique}
The inflation technique~\cite{WolfeSpekkensFritz+2019} is an important tool that allows to set constraints on the correlations that can arise in any network (up to computational complexity issues). 
Intuitively, we are concerned with the hypothetical situation where one has access to multiple copies of the sources and measurement devices that compose the network and can rearrange them in different configurations. 
Its core idea is to explore simple (linear) conditions of this inflated network that ultimately translate to polynomial inequalities on the observable probabilities. 
It has been proven in~\cite{NavascuesWolfe+2020+70+91} the existence of a hierarchy of inflations that asymptotically converges to the classical set of correlations of any network and a test of compatibility of a given level of this hierarchy can be done via Linear Programming (LP)~\cite{BoydVandenberghe}. 
However, for each level $n$ of this hierarchy the memory resources required are superexponential on n. 
Here we will show how we can use e-separation to drastically decrease the number of variables for each compatibility test and prove the Bonet PR-Box incompatibility.

In parallel to the concept of Inflation, the concept of e-separation (extended d-separation) was introduced in~\cite{Evans2012}. 
Let $\mathbf{A}$,$\mathbf{B}$, $\mathbf{C}$ and $\mathbf{D}$ be disjoint sets of observable variables in a DAG $\mathcal{G}$.
We say the variables $\mathbf{A}$ and $\mathbf{B}$ are e-separated given $\mathbf{C}$ after deletion of $\mathbf{D}$ iff  $\mathbf{A}$ and $\mathbf{B}$ are d-separated by $\mathbf{C}$ in the DAG $\mathcal{G}^*$, which is the resulting DAG after removal of all the vertices in $\mathbf{D}$ from $\mathcal{G}$. If $\mathbf{D}=\emptyset$, e-separation recovers the notion of d-separation. 

The e-separation criterion is related to the idea of splitting nodes in a graph. We can define a node-splitting operation as follows. Given a graph $\mathcal{G}$ and a vertex $\mathbf{D}$ in the graph, the node splitting operation returns a new graph $\mathcal{G}^{\#}$ in which $\mathbf{D}$ is split into two vertices. One of the vertices, called $\mathbf{D}$, maintains all its causal parents in the original graph $\mathcal{G}$, thus having the same distribution as $\mathbf{D}$ in $\mathcal{G}$, but none of its outgoing edges.  
The other one instead, labeled $\mathbf{D}^{\#}$, will inherit all of the outgoing edges of $\mathbf{D}$ in the original graph, but none of its incoming ones. 
An example of the node splitting operation is illustrated in Fig.~\ref{fig:e-sep}. 
In~\cite{Evans2012} it was proven that $\mathbf{A}$ and $\mathbf{B}$ are e-separated given $\mathbf{C}$ after deletion of $\mathbf{D}$ iff $\mathbf{A}$ and $\mathbf{B}$ are d-separated by $\mathbf{C}$ and $\mathbf{D}^{\#}$ in $\mathcal{G}^{\#}$.

In the Evans scenario, this would be the case if we identify the nodes $\mathbf{A}=A$, $\mathbf{B}=C$, $\mathbf{C}=\emptyset$ and $\mathbf{D}=B$ (see Fig.~\ref{fig:e-sep}). 
The node splitting operation defines a set of conditions that follow from the fact that all nodes are copies of the node in the original DAG. 
Namely, for any probability distribution $p(a,b,c)$ compatible with Evans scenario there exists a probability distribution $q(a,b,c|b^{\#})$ compatible with the DAG on the right of Fig.~\ref{fig:e-sep} respecting
\begin{equation}
    q(a,b,c|b^{\#}=b)=p(a,b,c).
\end{equation}
After performing the node splitting operation we use the inflation technique in the resulting DAG. 
For simplicity let us show an example with a second-order inflation. 

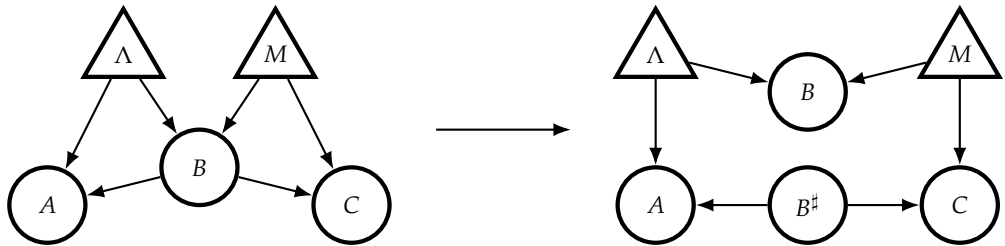
\begin{figure*}
\begin{tikzpicture}
    \node[var] (a) at (-2,0) {$A$};
    \node[var] (c) at (2,0) {$C$};
    %\node[var] (b) [right =of a] {$B$};
    \node[latent] (l) at (-1,2) {$\Lambda$};
    \node[latent] (g) at (1,2) {$M$};
    \node[var] (b) at (0,.5) {$B$};
    % Directed edges
    \path[dir] (l) edge (a) (l) edge (b); 
    \path[dir] (g) edge (b) (g) edge (c); 
    \path[dir] (b) edge (a) (b) edge (c);

	\node[var] (as) at (6,0) {$A$};
    \node[var] (cs) at (10,0) {$C$};
    %\node[var] (b) [right =of a] {$B$};
    \node[latent] (ls) at (6,2) {$\Lambda$};
    \node[latent] (gs) at (10,2) {$M$};
    \node[var] (bs) at (8,1.5) {$B$};
    \node[var] (bss) at (8,0) {$B^\sharp$};
    % Directed edges
    \path[dir] (ls) edge (as) (ls) edge (bs); 
    \path[dir] (gs) edge (bs) (gs) edge (cs); 
    \path[dir] (bss) edge (as) (bss) edge (cs);

	\node (f1) at (3,1) {};
	\node (f2) at (5,1) {};
    \path[dir] (f1) edge (f2);
\end{tikzpicture}
\caption{\textbf{Node splitting operation.} In the node splitting operation the DAG $\mathcal{G}$ on the left, in this case representing the Evans scenario, is associated with the DAG $\mathcal{G}^\#$ on the right, where one of the nodes, the node $B$, has been split into two nodes $B$ and $B^\#$ retaining respectively the incoming and outgoing edges of the original node.
The DAG $\mathcal{G}^\#$ is useful to visualize the e-separation condition between the nodes $A$ and $C$ deleting $B$, which here corresponds to the d-separation condition $\dsep{A}{C}{B^\#}$.}
\label{fig:e-sep}
\end{figure*}

Considering two independent copies of each latent variable $\lambda_i$, $\mu_i$, $B_i^{\#}$ for $i=1,2$ and $A_i=A(\lambda_i,B^{\#}_i)$, $C_j=C(\mu_j,B^{\#}_j) $ and $B_{ij}=B(\lambda_i,\mu_j)$, give us a joint distribution on all observed variables, since the $B^{\#}$ variables are exogenous we can interpret them as inputs. 
Suppose that $q$ is compatible with the DAG on the right of Fig~\ref{fig:e-sep}, i.e. 
\begin{equation}
    q(a,b,c|b^{\#})=\sum_{\lambda, \mu }q(\lambda)q(\mu)q(a|\lambda ,b^{\#})q(b|\lambda ,\mu)q(c|\mu, b^{\#}) \, ,
\end{equation}
then there exists a joint distribution on the inflated DAG respecting 
\begin{multline}
\label{eq:markov2nd-order}
     q'(a_1,a_2,b_{11},b_{12},b_{21},b_{22},c_1,c_2|b_1^{\#},b_2^{\#})=\\
     =\sum_{\lambda_1, \mu_1,\lambda_2, \mu_2}q(\mu_1)q(\lambda_1)q(\mu_2) q(\lambda_2)q(a_1|\lambda_1 b_1^{\#})q(a_2|\lambda_2 b_2^{\#}) \times \\ 
     \times \prod_{i,j=1,2}q(b_{ij}|\lambda_i \mu_j)q(c_1|m_1 b^{\#})q(c_2|\mu_2 b_2^{\#}).
 \end{multline}
Although we cannot test~\eqref{eq:markov2nd-order} directly, the equation imposes linear constraints on the distribution $q$. 
Switching the labels $(\mu_1 \leftrightarrow \mu_2, \lambda_1 \leftrightarrow \lambda_2, a_1\leftrightarrow a_2,c_1\leftrightarrow c_2, b_{11}\leftrightarrow b_{22}, b_{12}\leftrightarrow b_{21}, b_1^{\#} \leftrightarrow  b_2^{\#})$ leaves equation~\eqref{eq:markov2nd-order} unchanged. 
This implies 
\begin{multline}
\label{eq: symmetryinflation}
    q'(a_1,a_2,b_{11},b_{12},b_{21},b_{22},c_1,c_2|b_1^{\#}b_2^{\#})=\\
        = q'(a_2,a_1,b_{22},b_{21},b_{12},b_{11},c_2,c_1|b_2^{\#}b_1^{\#})\, .
\end{multline}

We may also consider the same expression when $b^{\#}_1=b^{\#}_2=b^{\#}$, switching the labels $(\lambda_1\leftrightarrow \lambda_2, a_1 \leftrightarrow a_2, b_{11} \leftrightarrow b_{21}, b_{12}\leftrightarrow b_{22} )$ or $(\mu_1\leftrightarrow \mu_2, c_1 \leftrightarrow c_2, b_{11} \leftrightarrow b_{12}, b_{21}\leftrightarrow b_{22} )$ yields 
\begin{equation}
\label{eq: symmetryinflation1}
    \begin{aligned}
        &q'(a_1,a_2,b_{11},b_{12},b_{21},b_{22},c_1,c_2|b_1^{\#}=b_2^{\#}=b^{\#})=\\
        &=q'(a_2,a_1,b_{21},b_{22},b_{11},b_{12},c_1,c_2|b_1^{\#}=b_2^{\#}=b^{\#})\\
        &=q'(a_1,a_2,b_{12},b_{11},b_{22},b_{21},c_2,c_1|b_1^{\#}=b_2^{\#}=b^{\#}),
    \end{aligned}
    \end{equation}
 and it has to satisfy some trivial conditions such as normalization, non-negativity and that the outputs $a_i,c_i$ are independent of inputs $b_{j\neq i}^{\#}$. Finally, condition~\eqref{eq:markov2nd-order} tells us that $q'$ recovers $q$ via marginalization, but only some entries of $q$ are available since we, initially, only have $p$. This yields a set of linear conditions
 \begin{equation}
 \begin{aligned}
     \label{Marg-cond.}
     q'(\{a_i,b_{ii},c_i|b_i^{\#}=b_{ii}\})&=\prod_i q(a_i,b_{ii},c_i|b_i^{\#}=b_{ii})=\\
     &=\prod_i p(a_i,b_{ii},c_i).
 \end{aligned}
 \end{equation}
 An analogous procedure also follows for the general case where, in the n-th order, we would consider $n$ independent copies of $\Lambda$, $M$, $B^{\#}$ and $n^2$ copies of $A( \lambda ,B^{\#})$, $B(\lambda,\mu)$ and $C(\mu,B^{\#})$. We need to consider all the relabelings that leave the Markov model of the resulting network invariant and recover n-degree polynomials over the diagonal distribution $q_n'(\{A_i,B_{ii},C_i|B_i^{\#}=B_{ii}\})$.
 For the Bonet PR-Box \ref{eq:NS_Box} we are able to find incompatibility for $n=3$.
 It's also important to notice that e-separation + 2nd order Inflation yields, for our case where $|B|=|C|=2$ and $|A|=3$,  an LP with 2.304 variables and for 3rd level yields 884.736 variables, while the hierarchy detailed in \cite{NavascuesWolfe+2020+70+91} yields 20.736 variables for 2nd order and 23.887.872 variables for 3rd order inflation, which would be quite challenging even for an LP.

We can translate the unfeasible status of our certification in terms of a witness 
%i.e. compatibility inequality that is violated for our candidate distribution but is valid for all distributions compatible with the scenario, 
via Farkas' lemma~\cite{BoydVandenberghe}, which states that either the linear system $Ax=b$, $x\geq 0$ has a solution or $A^Ty\leq 0$ has a solution with $b^Ty<0$. 
Thus, if our LP certification has no solution for a given $p$, there exists a solution $y_{op}$ so that the symbolic expression 
\begin{equation}
    b^Ty_{op}\geq 0
\end{equation}
can be understood as a causal inequality, where $b=b(p)$ is a vector that has entries of all monomials of degree at most $n$ due to~\eqref{Marg-cond.}. 
Applying this procedure directly on our Bonet PR-Box with $p_A(0) = p_A(2) = 10/21$ and $p_A(1) = 1/21$ we are able to retrieve such witness in terms of cubic monomials of $p(a,b,c)$. 
For simplicity, we could make the mild assumption $p(c=f(a,b)+b)=1$, which can be guaranteed under classical models and is true for our candidate distribution \ref{eq:NS_Box}, to write down the inequality in Appendix~\ref{app: IneqEvansInflation} 

Our candidate distribution reaches a violation of $\beta_{PR}=3.2 \times 10^{-3}$ and the inequality gives us a white noise visibility of $v=0.9984$. 
We were not able to find any quantum violations of this witness.

\section{Deriving theory independent constraints}
\label{sec: Theory_independent_constraints}
So far, we have focused on the derivation of constraints valid for a classical description of Evans causal structure. In some cases, however, it might be interesting to have constraints valid for generalized probabilistic theories (GPTs), that is, constraints reflecting the topology of the causal network rather than the nature (classical, quantum or even post-quantum) of the sources. This theory independent constraints can be seen as genuine witnesses of the topology of the network and have already been derived for certain classes of causal networks \cite{kela2019semidefinite,aaberg2020semidefinite,beigi2021covariance}, most prominently the triangle network \cite{GisinBancalCai,renou2019limits}. In the following, we propose a general route for deriving such theory-independent constraints for the Evans causal structure and derive one specific inequality to demonstrate the method.

A crucial difference between a classical and a general probabilistic description of a causal network stems from the fact that the local changes performed by any part in the network cannot be arbitrary as, for instance, they may possess resources that cannot be perfectly cloned. That is, differently from the classical case, we can only consider non-fanout inflations, as we cannot broadcast information that comes from non-classical latent variables. To that respect, a notable difference in the Evans scenario is that, differently from the triangle, it is possible to include tripartite correlations from non-fanout inflations. 

Looking from this perspective, we can see that in the node splitting operation, discussed in the last section, the inward arrows to $A$ and $C$ do not change. Therefore, we can see this operation as a local modification in the network topology performed by $B$. Thus, e-separation should also follow for all no-signaling behaviors. 
In fact, it was shown in \cite{Evans2012,Finkelstein2021EntropicIC} that e-separation alone can impose constraints on GPT behaviors of the network. However, we can recursively employ these local operations in our network. Namely, after some part performs a local operation we can end up with a new network in which the new parts are independent and are also allowed to locally change the network structure. 
Notice that, in the last section, we used e-separation to derive constraints on the classical set of correlations, these constraints are violated by the PR-Box due to the explicit use of inflationary fanouts that only follow for classical latent variables.

Inspired by the hexagon inflation in \cite{GisinBancalCai}, Bob can set up this new experiment such that the laboratory of $Bob^{\#}$ is very far away and $Bob^{\#}$ may adopt a similar strategy, again changing the network topology and this should be oblivious to Alice, Bob and Charlie (see Fig.\ref{fig: InflationI}). Notice that if we were to consider two independent copies of the node $B^{\#}$, which we could, we would retrieve the hexagon by interpreting the $B^{\#}$ nodes as observable shared randomness between the parts. In our case, the d-separation and marginal independence imply non-linear constraints involving non-observable terms, posing a more challenging non-linear quantifier elimination procedure. Fortunately, we do not need to consider independent copies of the $B^{\#}$ variable, since it is an observable variable and therefore we can assume it to be classical regardless of the resources the parts share. This means that the information sent to $A$ and $C$ can be perfectly cloned, thus we may broadcast it throughout the network. With these remarks, the conditions implied are 
\begin{equation}
    q(a,a',b=b',c,c'|b^{\#}=b)=p(a,b,c)p(a',b'=b,c'),
\end{equation} 

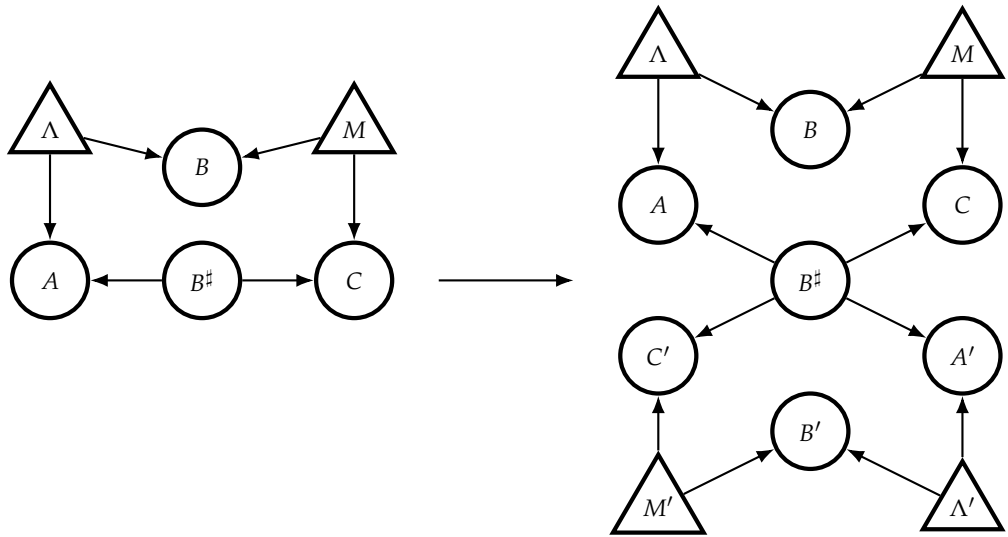
\begin{figure*}
\begin{center}
    \begin{tikzpicture}

	\node[var] (as) at (-2,0) {$A$};
    \node[var] (cs) at (2,0) {$C$};
    \node[latent] (ls) at (-2,2) {$\Lambda$};
    \node[latent] (gs) at (2,2) {$M$};
    \node[var] (bs) at (0,1.5) {$B$};
    \node[var] (bss) at (0,0) {$B^\sharp$};
    % Directed edges
    \path[dir] (ls) edge (as) (ls) edge (bs); 
    \path[dir] (gs) edge (bs) (gs) edge (cs); 
    \path[dir] (bss) edge (as) (bss) edge (cs);
    
	\node (f1) at (3,0) {};
	\node (f2) at (5,0) {};
    \path[dir] (f1) edge (f2);

	\node[var] (ai1) at (6,1) {$A$};
    \node[var] (ci1) at (6,-1) {$C'$};
	\node[var] (ai2) at (10,-1) {$A'$};
    \node[var] (ci2) at (10,1) {$C$};
    \node[latent] (li1) at (6,3) {$\Lambda$};
    \node[latent] (gi1) at (6,-3) {$M'$};
    \node[latent] (li2) at (10,-3) {$\Lambda'$};
    \node[latent] (gi2) at (10,3) {$M$};
    \node[var] (bi1) at (8,2) {$B$};
    \node[var] (bi2) at (8,-2) {$B'$};
    \node[var] (bis) at (8,0) {$B^\sharp$};
    % Directed edges
    \path[dir] (bis) edge (ai1) (bis) edge (ci1) (bis) edge (ai2) (bis) edge (ci2); 
    \path[dir] (li1) edge (ai1) (li1) edge (bi1);
    \path[dir] (li2) edge (ai2) (li2) edge (bi2);
    \path[dir] (gi1) edge (ci1) (gi1) edge (bi2);
    \path[dir] (gi2) edge (ci2) (gi2) edge (bi1);
\end{tikzpicture}
        \caption{\textbf{Inflation of the Evans scenario - } In order to derive the topological constraints of the scenario on the left we consider the non-fanout inflation on the right. Importantly, the operations needed to bring the original Evans scenario to the node splitted DAG and then to the inflated causal structure follow for all GPTs and can be performed locally by the parts that conform the network. }
        \label{fig: InflationI}
\end{center}
\end{figure*}
as well as  $$q(b,b'|b^{\#})=q(b,b')\quad \forall b^{\#},$$ and normalization conditions.
These equality constraints together with the non-negativity conditions will form a set of linear inequalities on the variables $q$. Since the quadratic terms of $p$ may be treated like symbolic constants, we may perform Fourier Motzkin elimination to arrive at quadratic inequalities that should follow for all no-signaling behaviors. 

The procedure outlined above is quite general and could be used to derive all constraints following a given inflation. However, it can be very costly and out of computational reach even for seemingly simple scenarios such as Evans. To circumvent that, similarly to the approach discussed in the last section, given a behavior from the Evans scenario, we may employ a linear program (LP) to detect incompatibility between the distribution and the set of inequalities. That is, given a distribution $p(a,b,c)$ we solve:
\begin{align}\nonumber
\text{max}_{q} \quad &1\\
\text{s.t.} \quad &q(a,a',b=b',c,c'|b^{\#}=b)=p(a,b,c)p(a',b'=b,c')\label{LP}\\
&q(b,b'|b^{\#}) = q(b, b')\quad\forall \; b^{\#} \nonumber\\
&\sum q(a,a',b,b',c,c'|b^{\#})=1 \quad \forall b^{\#}\nonumber\\
& q(a,a',b,b',c,c'|b^{\#}) \geq 0 \quad \forall a,a',b,b',c,c', b^{\#}\nonumber.
\end{align}
If this LP is feasible, the behavior $p(a,b,c)$ is compatible with the constraints. Otherwise, it is not. If we have a non-feasible distribution we can then rely on the dual of the LP \eqref{LP} to easily find the corresponding witness (a non-linear inequality) detecting it.

Using this idea, we can show that the distribution $p_{unfeasible}(a,b,c)$ given by $p(0,0,0)=p(1,0,1) = 1/2$ is incompatible with the Evans topology, as it returns an unfeasible LP. From the dual problem, we can obtain a witness for the incompatibility, which reads as
\begin{align}
&p(0,0,0)^2 + p(0,0,1)p(0,0,0) + p(0,0,1)^2\nonumber \\
&+ 4 p(1,0,1)p(0,0,0) + p(1,0,0)^2+\nonumber\\ 
&+ p(1,0,1)p(1,0,0) + p(1,0,1)^2 \le 1.
\end{align}
This theory-independent polynomial inequality is violated by the distribution $p(0,0,0)=p(1,0,1) = 1/2$ as the right-hand side of it reaches the value $3/2$. Furthermore, this constraint turns out to be quite resistant to noise as can be seen by considering the mixed distribution $v p_{unfeasible}+(1-v)p_W$ ($p_W$ being the white noise distribution where all binary outcomes have the same $1/8$ probability). The threshold visibility for the inequality violation is $v_{crit}\approx 0,7398$.

\section{Discussion}
Bell’s theorem is the most stringent notion of nonclassicality since it only relies on causal assumptions about the experimental apparatus but remains agnostic of any internal mechanisms or physical details of the measurement and state preparation devices. Given its foundational as well as applied consequences, generalizing such results to causal structures of increasing complexity and with different topologies is a timely and promising  research venue.

Within this context, the contribution of this paper is to analyze a simple yet elusive causal structure. It is known that there are three different classes of causal structures with three observable variables that can potentially give rise to a mismatch between classical and quantum predictions. Two of those are indeed known to support quantum non-classicality and have been the focus of much of the literature on the topic recently. On one side we have the instrumental causal structure \cite{pearl1995testability,chaves2018quantum}, a scenario of crucial relevance for estimating causal influences \cite{balke1997bounds} and that, in particular, brings to light the relevance of considering interventions, rather than pure passive observations, in order to reveal the non-classicality of a given quantum process \cite{gachechiladze2020quantifying,agresti2022experimental}. On the other side, we have the triangle scenario \cite{fritz2012beyond,renou2019genuine} with the central feature of, differently from Bell's theorem, not requiring external inputs choosing between the measurement of different observables and solely resorting on the independence of sources for leading to non-classical behaviors \cite{polino2022experimental}. Here we focus on the third of these causal structures, the Evans scenario \cite{evans2016graphs}, a sort of blend between the instrumental and triangle structures since it involves independent sources and no external inputs but does presume direct causal influences between the measurement outcomes.

Being akin to the entanglement swapping experiment \cite{zukowski1993event} it is natural to expect that the correlations obtained by considering a Bell state measurement performed on shared Bell states should lead to non-classical correlations. Quite the opposite, we proved that such an experiment does have a classical causal model for any dimension. So, in order to prove non-classicality, we explore different venues, establishing the connections of Evans with the instrumental and bilocality scenarios. In particular, we prove that non-classicality in the bilocality structure is a necessary but not sufficient condition for non-classicality in the Evans case. Unfortunately, however, this bridge with well-known scenarios proved not enough to find non-classicality in the Evans structure, motivating us to search for alternative methods to study the set of correlations entailed by it.

With that in mind, we discussed the features of the Gurobi optimizer \cite{bixby2007gurobi}, capable of handling non-convex quadratic optimization problems and being perfectly fitted to analyze the Evans causal structure. With the Gurobi optimizer, we could not only numerically prove that the analogous of a PR-box in the instrumental scenario does lead to non-classical correlations in the Evans scenario but also recover numerical witness for it. Going beyond this numerical study we combined the inflation \cite{WolfeSpekkensFritz+2019} and e-separation \cite{Evans2012} techniques to derive polynomial Bell inequalities that can be violated, an unambiguous proof that the Evans causal structure also can support non-classical correlations, even though they are of a post-quantum nature. Finally, we showed how to derive theory-independent constraints for the Evans scenario, that is, inequalities that reflect the topology of the network, similar to what has been done to the triangle scenario \cite{GisinBancalCai}.

The natural next step is to finally resolve whether quantum non-classical behavior is or is not possible in the Evans scenario. It is worth remarking that here we focused only on the observed distribution $p(a,b,c)$. However, as shown recently for the instrumental scenario \cite{gachechiladze2020quantifying}, interventions could be a possible way to unveil the non-classicality, that is, resorting to interventional data of the form $p(a,c\vert do(b))$, where $do(b)$ refers to interventions in the central node of the network.  Either way, the resolution of this problem will certainly be of relevance to the larger question of understanding what are the crucial ingredients of the emergence of non-classical behavior in causal structures of growing size and complexity.
\section*{ACKNOWLEDGMENTS}

We thank Elie Wolfe for fruitful discussions. This work was supported by the Serrapilheira Institute (Grant No. Serra-1708-15763), by the Simons Foundation (Grant Number 884966, AF), the Brazilian National Council for Scientific and Technological Development (CNPq) via the National Institute for Science and Technology on Quantum Information (INCT-IQ) and Grant No. and 307295/2020-6, and the São Paulo Research Foundation FAPESP (Grant Nos. 2018/07258-7, 2022/03792-4).

\bibliography{ref}

\begin{thebibliography}{68}
\expandafter\ifx\csname natexlab\endcsname\relax\def\natexlab#1{#1}\fi
\expandafter\ifx\csname bibnamefont\endcsname\relax
  \def\bibnamefont#1{#1}\fi
\expandafter\ifx\csname bibfnamefont\endcsname\relax
  \def\bibfnamefont#1{#1}\fi
\expandafter\ifx\csname citenamefont\endcsname\relax
  \def\citenamefont#1{#1}\fi
\expandafter\ifx\csname url\endcsname\relax
  \def\url#1{\texttt{#1}}\fi
\expandafter\ifx\csname urlprefix\endcsname\relax\def\urlprefix{URL }\fi
\providecommand{\bibinfo}[2]{#2}
\providecommand{\eprint}[2][]{\url{#2}}

\bibitem[{\citenamefont{Bell}(1964)}]{bell1964einstein}
\bibinfo{author}{\bibfnamefont{J.~S.} \bibnamefont{Bell}},
  \bibinfo{journal}{Physics Physique Fizika} \textbf{\bibinfo{volume}{1}},
  \bibinfo{pages}{195} (\bibinfo{year}{1964}).

\bibitem[{\citenamefont{Brunner et~al.}(2014)\citenamefont{Brunner, Cavalcanti,
  Pironio, Scarani, and Wehner}}]{brunner2014bell}
\bibinfo{author}{\bibfnamefont{N.}~\bibnamefont{Brunner}},
  \bibinfo{author}{\bibfnamefont{D.}~\bibnamefont{Cavalcanti}},
  \bibinfo{author}{\bibfnamefont{S.}~\bibnamefont{Pironio}},
  \bibinfo{author}{\bibfnamefont{V.}~\bibnamefont{Scarani}}, \bibnamefont{and}
  \bibinfo{author}{\bibfnamefont{S.}~\bibnamefont{Wehner}},
  \bibinfo{journal}{Reviews of Modern Physics} \textbf{\bibinfo{volume}{86}},
  \bibinfo{pages}{419} (\bibinfo{year}{2014}).

\bibitem[{\citenamefont{Pironio et~al.}(2016)\citenamefont{Pironio, Scarani,
  and Vidick}}]{pironio2016focus}
\bibinfo{author}{\bibfnamefont{S.}~\bibnamefont{Pironio}},
  \bibinfo{author}{\bibfnamefont{V.}~\bibnamefont{Scarani}}, \bibnamefont{and}
  \bibinfo{author}{\bibfnamefont{T.}~\bibnamefont{Vidick}},
  \bibinfo{journal}{New J. Phys} \textbf{\bibinfo{volume}{18}},
  \bibinfo{pages}{100202} (\bibinfo{year}{2016}).

\bibitem[{\citenamefont{Fitzsimons et~al.}(2015)\citenamefont{Fitzsimons,
  Jones, and Vedral}}]{fitzsimons2015quantum}
\bibinfo{author}{\bibfnamefont{J.~F.} \bibnamefont{Fitzsimons}},
  \bibinfo{author}{\bibfnamefont{J.~A.} \bibnamefont{Jones}}, \bibnamefont{and}
  \bibinfo{author}{\bibfnamefont{V.}~\bibnamefont{Vedral}},
  \bibinfo{journal}{Scientific reports} \textbf{\bibinfo{volume}{5}},
  \bibinfo{pages}{18281} (\bibinfo{year}{2015}).

\bibitem[{\citenamefont{Ried et~al.}(2015)\citenamefont{Ried, Agnew, Vermeyden,
  Janzing, Spekkens, and Resch}}]{ried2015quantum}
\bibinfo{author}{\bibfnamefont{K.}~\bibnamefont{Ried}},
  \bibinfo{author}{\bibfnamefont{M.}~\bibnamefont{Agnew}},
  \bibinfo{author}{\bibfnamefont{L.}~\bibnamefont{Vermeyden}},
  \bibinfo{author}{\bibfnamefont{D.}~\bibnamefont{Janzing}},
  \bibinfo{author}{\bibfnamefont{R.~W.} \bibnamefont{Spekkens}},
  \bibnamefont{and} \bibinfo{author}{\bibfnamefont{K.~J.} \bibnamefont{Resch}},
  \bibinfo{journal}{Nature Physics} \textbf{\bibinfo{volume}{11}},
  \bibinfo{pages}{414} (\bibinfo{year}{2015}).

\bibitem[{\citenamefont{Fritz}(2016)}]{fritz2016beyond}
\bibinfo{author}{\bibfnamefont{T.}~\bibnamefont{Fritz}},
  \bibinfo{journal}{Communications in Mathematical Physics}
  \textbf{\bibinfo{volume}{341}}, \bibinfo{pages}{391} (\bibinfo{year}{2016}).

\bibitem[{\citenamefont{Chaves et~al.}(2015)\citenamefont{Chaves, Majenz, and
  Gross}}]{chaves2015information}
\bibinfo{author}{\bibfnamefont{R.}~\bibnamefont{Chaves}},
  \bibinfo{author}{\bibfnamefont{C.}~\bibnamefont{Majenz}}, \bibnamefont{and}
  \bibinfo{author}{\bibfnamefont{D.}~\bibnamefont{Gross}},
  \bibinfo{journal}{Nature communications} \textbf{\bibinfo{volume}{6}},
  \bibinfo{pages}{1} (\bibinfo{year}{2015}).

\bibitem[{\citenamefont{Costa and Shrapnel}(2016)}]{costa2016quantum}
\bibinfo{author}{\bibfnamefont{F.}~\bibnamefont{Costa}} \bibnamefont{and}
  \bibinfo{author}{\bibfnamefont{S.}~\bibnamefont{Shrapnel}},
  \bibinfo{journal}{New Journal of Physics} \textbf{\bibinfo{volume}{18}},
  \bibinfo{pages}{063032} (\bibinfo{year}{2016}).

\bibitem[{\citenamefont{Allen et~al.}(2017)\citenamefont{Allen, Barrett,
  Horsman, Lee, and Spekkens}}]{allen2017quantum}
\bibinfo{author}{\bibfnamefont{J.-M.~A.} \bibnamefont{Allen}},
  \bibinfo{author}{\bibfnamefont{J.}~\bibnamefont{Barrett}},
  \bibinfo{author}{\bibfnamefont{D.~C.} \bibnamefont{Horsman}},
  \bibinfo{author}{\bibfnamefont{C.~M.} \bibnamefont{Lee}}, \bibnamefont{and}
  \bibinfo{author}{\bibfnamefont{R.~W.} \bibnamefont{Spekkens}},
  \bibinfo{journal}{Physical Review X} \textbf{\bibinfo{volume}{7}},
  \bibinfo{pages}{031021} (\bibinfo{year}{2017}).

\bibitem[{\citenamefont{{\AA}berg et~al.}(2020)\citenamefont{{\AA}berg, Nery,
  Duarte, and Chaves}}]{aaberg2020semidefinite}
\bibinfo{author}{\bibfnamefont{J.}~\bibnamefont{{\AA}berg}},
  \bibinfo{author}{\bibfnamefont{R.}~\bibnamefont{Nery}},
  \bibinfo{author}{\bibfnamefont{C.}~\bibnamefont{Duarte}}, \bibnamefont{and}
  \bibinfo{author}{\bibfnamefont{R.}~\bibnamefont{Chaves}},
  \bibinfo{journal}{Physical Review Letters} \textbf{\bibinfo{volume}{125}},
  \bibinfo{pages}{110505} (\bibinfo{year}{2020}).

\bibitem[{\citenamefont{Wolfe et~al.}(2021)\citenamefont{Wolfe,
  Pozas-Kerstjens, Grinberg, Rosset, Ac{\'\i}n, and
  Navascu{\'e}s}}]{wolfe2021quantum}
\bibinfo{author}{\bibfnamefont{E.}~\bibnamefont{Wolfe}},
  \bibinfo{author}{\bibfnamefont{A.}~\bibnamefont{Pozas-Kerstjens}},
  \bibinfo{author}{\bibfnamefont{M.}~\bibnamefont{Grinberg}},
  \bibinfo{author}{\bibfnamefont{D.}~\bibnamefont{Rosset}},
  \bibinfo{author}{\bibfnamefont{A.}~\bibnamefont{Ac{\'\i}n}},
  \bibnamefont{and}
  \bibinfo{author}{\bibfnamefont{M.}~\bibnamefont{Navascu{\'e}s}},
  \bibinfo{journal}{Physical Review X} \textbf{\bibinfo{volume}{11}},
  \bibinfo{pages}{021043} (\bibinfo{year}{2021}).

\bibitem[{\citenamefont{Ligthart et~al.}(2021)\citenamefont{Ligthart,
  Gachechiladze, and Gross}}]{ligthart2021convergent}
\bibinfo{author}{\bibfnamefont{L.~T.} \bibnamefont{Ligthart}},
  \bibinfo{author}{\bibfnamefont{M.}~\bibnamefont{Gachechiladze}},
  \bibnamefont{and} \bibinfo{author}{\bibfnamefont{D.}~\bibnamefont{Gross}},
  \bibinfo{journal}{arXiv preprint arXiv:2110.14659}  (\bibinfo{year}{2021}).

\bibitem[{\citenamefont{Tavakoli
  et~al.}(2021{\natexlab{a}})\citenamefont{Tavakoli, Pozas-Kerstjens, Renou
  et~al.}}]{tavakoli2021bell}
\bibinfo{author}{\bibfnamefont{A.}~\bibnamefont{Tavakoli}},
  \bibinfo{author}{\bibfnamefont{A.}~\bibnamefont{Pozas-Kerstjens}},
  \bibinfo{author}{\bibfnamefont{M.-O.} \bibnamefont{Renou}},
  \bibnamefont{et~al.}, \bibinfo{journal}{Reports on Progress in Physics}
  (\bibinfo{year}{2021}{\natexlab{a}}).

\bibitem[{\citenamefont{Branciard
  et~al.}(2010{\natexlab{a}})\citenamefont{Branciard, Gisin, and
  Pironio}}]{branciard2010characterizing}
\bibinfo{author}{\bibfnamefont{C.}~\bibnamefont{Branciard}},
  \bibinfo{author}{\bibfnamefont{N.}~\bibnamefont{Gisin}}, \bibnamefont{and}
  \bibinfo{author}{\bibfnamefont{S.}~\bibnamefont{Pironio}},
  \bibinfo{journal}{Physical review letters} \textbf{\bibinfo{volume}{104}},
  \bibinfo{pages}{170401} (\bibinfo{year}{2010}{\natexlab{a}}).

\bibitem[{\citenamefont{Fritz}(2012)}]{fritz2012beyond}
\bibinfo{author}{\bibfnamefont{T.}~\bibnamefont{Fritz}}, \bibinfo{journal}{New
  Journal of Physics} \textbf{\bibinfo{volume}{14}}, \bibinfo{pages}{103001}
  (\bibinfo{year}{2012}).

\bibitem[{\citenamefont{Tavakoli et~al.}(2014)\citenamefont{Tavakoli,
  Skrzypczyk, Cavalcanti, and Ac{\'\i}n}}]{tavakoli2014nonlocal}
\bibinfo{author}{\bibfnamefont{A.}~\bibnamefont{Tavakoli}},
  \bibinfo{author}{\bibfnamefont{P.}~\bibnamefont{Skrzypczyk}},
  \bibinfo{author}{\bibfnamefont{D.}~\bibnamefont{Cavalcanti}},
  \bibnamefont{and}
  \bibinfo{author}{\bibfnamefont{A.}~\bibnamefont{Ac{\'\i}n}},
  \bibinfo{journal}{Physical Review A} \textbf{\bibinfo{volume}{90}},
  \bibinfo{pages}{062109} (\bibinfo{year}{2014}).

\bibitem[{\citenamefont{Andreoli et~al.}(2017)\citenamefont{Andreoli, Carvacho,
  Santodonato, Chaves, and Sciarrino}}]{andreoli2017maximal}
\bibinfo{author}{\bibfnamefont{F.}~\bibnamefont{Andreoli}},
  \bibinfo{author}{\bibfnamefont{G.}~\bibnamefont{Carvacho}},
  \bibinfo{author}{\bibfnamefont{L.}~\bibnamefont{Santodonato}},
  \bibinfo{author}{\bibfnamefont{R.}~\bibnamefont{Chaves}}, \bibnamefont{and}
  \bibinfo{author}{\bibfnamefont{F.}~\bibnamefont{Sciarrino}},
  \bibinfo{journal}{New Journal of Physics} \textbf{\bibinfo{volume}{19}},
  \bibinfo{pages}{113020} (\bibinfo{year}{2017}).

\bibitem[{\citenamefont{Renou et~al.}(2019{\natexlab{a}})\citenamefont{Renou,
  B{\"a}umer, Boreiri, Brunner, Gisin, and Beigi}}]{renou2019genuine}
\bibinfo{author}{\bibfnamefont{M.-O.} \bibnamefont{Renou}},
  \bibinfo{author}{\bibfnamefont{E.}~\bibnamefont{B{\"a}umer}},
  \bibinfo{author}{\bibfnamefont{S.}~\bibnamefont{Boreiri}},
  \bibinfo{author}{\bibfnamefont{N.}~\bibnamefont{Brunner}},
  \bibinfo{author}{\bibfnamefont{N.}~\bibnamefont{Gisin}}, \bibnamefont{and}
  \bibinfo{author}{\bibfnamefont{S.}~\bibnamefont{Beigi}},
  \bibinfo{journal}{Physical review letters} \textbf{\bibinfo{volume}{123}},
  \bibinfo{pages}{140401} (\bibinfo{year}{2019}{\natexlab{a}}).

\bibitem[{\citenamefont{Tavakoli
  et~al.}(2021{\natexlab{b}})\citenamefont{Tavakoli, Gisin, and
  Branciard}}]{tavakoli2021bilocal}
\bibinfo{author}{\bibfnamefont{A.}~\bibnamefont{Tavakoli}},
  \bibinfo{author}{\bibfnamefont{N.}~\bibnamefont{Gisin}}, \bibnamefont{and}
  \bibinfo{author}{\bibfnamefont{C.}~\bibnamefont{Branciard}},
  \bibinfo{journal}{Physical review letters} \textbf{\bibinfo{volume}{126}},
  \bibinfo{pages}{220401} (\bibinfo{year}{2021}{\natexlab{b}}).

\bibitem[{\citenamefont{Chaves et~al.}(2021)\citenamefont{Chaves, Moreno,
  Polino, Poderini, Agresti, Suprano, Barros, Carvacho, Wolfe, Canabarro
  et~al.}}]{chaves2021causal}
\bibinfo{author}{\bibfnamefont{R.}~\bibnamefont{Chaves}},
  \bibinfo{author}{\bibfnamefont{G.}~\bibnamefont{Moreno}},
  \bibinfo{author}{\bibfnamefont{E.}~\bibnamefont{Polino}},
  \bibinfo{author}{\bibfnamefont{D.}~\bibnamefont{Poderini}},
  \bibinfo{author}{\bibfnamefont{I.}~\bibnamefont{Agresti}},
  \bibinfo{author}{\bibfnamefont{A.}~\bibnamefont{Suprano}},
  \bibinfo{author}{\bibfnamefont{M.~R.} \bibnamefont{Barros}},
  \bibinfo{author}{\bibfnamefont{G.}~\bibnamefont{Carvacho}},
  \bibinfo{author}{\bibfnamefont{E.}~\bibnamefont{Wolfe}},
  \bibinfo{author}{\bibfnamefont{A.}~\bibnamefont{Canabarro}},
  \bibnamefont{et~al.}, \bibinfo{journal}{PRX Quantum}
  \textbf{\bibinfo{volume}{2}}, \bibinfo{pages}{040323} (\bibinfo{year}{2021}).

\bibitem[{\citenamefont{Pozas-Kerstjens
  et~al.}(2022)\citenamefont{Pozas-Kerstjens, Gisin, and
  Tavakoli}}]{pozas2022full}
\bibinfo{author}{\bibfnamefont{A.}~\bibnamefont{Pozas-Kerstjens}},
  \bibinfo{author}{\bibfnamefont{N.}~\bibnamefont{Gisin}}, \bibnamefont{and}
  \bibinfo{author}{\bibfnamefont{A.}~\bibnamefont{Tavakoli}},
  \bibinfo{journal}{Physical review letters} \textbf{\bibinfo{volume}{128}},
  \bibinfo{pages}{010403} (\bibinfo{year}{2022}).

\bibitem[{\citenamefont{Carvacho et~al.}(2017)\citenamefont{Carvacho, Andreoli,
  Santodonato, Bentivegna, Chaves, and Sciarrino}}]{carvacho2017experimental}
\bibinfo{author}{\bibfnamefont{G.}~\bibnamefont{Carvacho}},
  \bibinfo{author}{\bibfnamefont{F.}~\bibnamefont{Andreoli}},
  \bibinfo{author}{\bibfnamefont{L.}~\bibnamefont{Santodonato}},
  \bibinfo{author}{\bibfnamefont{M.}~\bibnamefont{Bentivegna}},
  \bibinfo{author}{\bibfnamefont{R.}~\bibnamefont{Chaves}}, \bibnamefont{and}
  \bibinfo{author}{\bibfnamefont{F.}~\bibnamefont{Sciarrino}},
  \bibinfo{journal}{Nature communications} \textbf{\bibinfo{volume}{8}},
  \bibinfo{pages}{1} (\bibinfo{year}{2017}).

\bibitem[{\citenamefont{Saunders et~al.}(2017)\citenamefont{Saunders, Bennet,
  Branciard, and Pryde}}]{saunders2017experimental}
\bibinfo{author}{\bibfnamefont{D.~J.} \bibnamefont{Saunders}},
  \bibinfo{author}{\bibfnamefont{A.~J.} \bibnamefont{Bennet}},
  \bibinfo{author}{\bibfnamefont{C.}~\bibnamefont{Branciard}},
  \bibnamefont{and} \bibinfo{author}{\bibfnamefont{G.~J.} \bibnamefont{Pryde}},
  \bibinfo{journal}{Science advances} \textbf{\bibinfo{volume}{3}},
  \bibinfo{pages}{e1602743} (\bibinfo{year}{2017}).

\bibitem[{\citenamefont{Sun et~al.}(2019)\citenamefont{Sun, Jiang, Bai, Zhang,
  Li, Jiang, Zhang, You, Chen, Wang et~al.}}]{sun2019experimental}
\bibinfo{author}{\bibfnamefont{Q.-C.} \bibnamefont{Sun}},
  \bibinfo{author}{\bibfnamefont{Y.-F.} \bibnamefont{Jiang}},
  \bibinfo{author}{\bibfnamefont{B.}~\bibnamefont{Bai}},
  \bibinfo{author}{\bibfnamefont{W.}~\bibnamefont{Zhang}},
  \bibinfo{author}{\bibfnamefont{H.}~\bibnamefont{Li}},
  \bibinfo{author}{\bibfnamefont{X.}~\bibnamefont{Jiang}},
  \bibinfo{author}{\bibfnamefont{J.}~\bibnamefont{Zhang}},
  \bibinfo{author}{\bibfnamefont{L.}~\bibnamefont{You}},
  \bibinfo{author}{\bibfnamefont{X.}~\bibnamefont{Chen}},
  \bibinfo{author}{\bibfnamefont{Z.}~\bibnamefont{Wang}}, \bibnamefont{et~al.},
  \bibinfo{journal}{Nature Photonics} \textbf{\bibinfo{volume}{13}},
  \bibinfo{pages}{687} (\bibinfo{year}{2019}).

\bibitem[{\citenamefont{Poderini et~al.}(2020)\citenamefont{Poderini, Agresti,
  Marchese, Polino, Giordani, Suprano, Valeri, Milani, Spagnolo, Carvacho
  et~al.}}]{poderini2020experimental}
\bibinfo{author}{\bibfnamefont{D.}~\bibnamefont{Poderini}},
  \bibinfo{author}{\bibfnamefont{I.}~\bibnamefont{Agresti}},
  \bibinfo{author}{\bibfnamefont{G.}~\bibnamefont{Marchese}},
  \bibinfo{author}{\bibfnamefont{E.}~\bibnamefont{Polino}},
  \bibinfo{author}{\bibfnamefont{T.}~\bibnamefont{Giordani}},
  \bibinfo{author}{\bibfnamefont{A.}~\bibnamefont{Suprano}},
  \bibinfo{author}{\bibfnamefont{M.}~\bibnamefont{Valeri}},
  \bibinfo{author}{\bibfnamefont{G.}~\bibnamefont{Milani}},
  \bibinfo{author}{\bibfnamefont{N.}~\bibnamefont{Spagnolo}},
  \bibinfo{author}{\bibfnamefont{G.}~\bibnamefont{Carvacho}},
  \bibnamefont{et~al.}, \bibinfo{journal}{Nature communications}
  \textbf{\bibinfo{volume}{11}}, \bibinfo{pages}{1} (\bibinfo{year}{2020}).

\bibitem[{\citenamefont{Carvacho et~al.}(2019)\citenamefont{Carvacho, Chaves,
  and Sciarrino}}]{carvacho2019perspective}
\bibinfo{author}{\bibfnamefont{G.}~\bibnamefont{Carvacho}},
  \bibinfo{author}{\bibfnamefont{R.}~\bibnamefont{Chaves}}, \bibnamefont{and}
  \bibinfo{author}{\bibfnamefont{F.}~\bibnamefont{Sciarrino}},
  \bibinfo{journal}{EPL (Europhysics Letters)} \textbf{\bibinfo{volume}{125}},
  \bibinfo{pages}{30001} (\bibinfo{year}{2019}).

\bibitem[{\citenamefont{Cao et~al.}(2022)\citenamefont{Cao, Renou, Zhang,
  Mass{\'e}, Coiteux-Roy, Liu, Huang, Li, Guo, and
  Wolfe}}]{cao2022experimental}
\bibinfo{author}{\bibfnamefont{H.}~\bibnamefont{Cao}},
  \bibinfo{author}{\bibfnamefont{M.-O.} \bibnamefont{Renou}},
  \bibinfo{author}{\bibfnamefont{C.}~\bibnamefont{Zhang}},
  \bibinfo{author}{\bibfnamefont{G.}~\bibnamefont{Mass{\'e}}},
  \bibinfo{author}{\bibfnamefont{X.}~\bibnamefont{Coiteux-Roy}},
  \bibinfo{author}{\bibfnamefont{B.-H.} \bibnamefont{Liu}},
  \bibinfo{author}{\bibfnamefont{Y.-F.} \bibnamefont{Huang}},
  \bibinfo{author}{\bibfnamefont{C.-F.} \bibnamefont{Li}},
  \bibinfo{author}{\bibfnamefont{G.-C.} \bibnamefont{Guo}}, \bibnamefont{and}
  \bibinfo{author}{\bibfnamefont{E.}~\bibnamefont{Wolfe}},
  \bibinfo{journal}{arXiv preprint arXiv:2201.12754}  (\bibinfo{year}{2022}).

\bibitem[{\citenamefont{Suprano et~al.}(2022)\citenamefont{Suprano, Poderini,
  Polino, Agresti, Carvacho, Canabarro, Wolfe, Chaves, and
  Sciarrino}}]{suprano2022experimental}
\bibinfo{author}{\bibfnamefont{A.}~\bibnamefont{Suprano}},
  \bibinfo{author}{\bibfnamefont{D.}~\bibnamefont{Poderini}},
  \bibinfo{author}{\bibfnamefont{E.}~\bibnamefont{Polino}},
  \bibinfo{author}{\bibfnamefont{I.}~\bibnamefont{Agresti}},
  \bibinfo{author}{\bibfnamefont{G.}~\bibnamefont{Carvacho}},
  \bibinfo{author}{\bibfnamefont{A.}~\bibnamefont{Canabarro}},
  \bibinfo{author}{\bibfnamefont{E.}~\bibnamefont{Wolfe}},
  \bibinfo{author}{\bibfnamefont{R.}~\bibnamefont{Chaves}}, \bibnamefont{and}
  \bibinfo{author}{\bibfnamefont{F.}~\bibnamefont{Sciarrino}},
  \bibinfo{journal}{arXiv preprint arXiv:2204.00388}  (\bibinfo{year}{2022}).

\bibitem[{\citenamefont{Polino et~al.}(2022)\citenamefont{Polino, Poderini,
  Rodari, Agresti, Suprano, Carvacho, Wolfe, Canabarro, Moreno, Milani
  et~al.}}]{polino2022experimental}
\bibinfo{author}{\bibfnamefont{E.}~\bibnamefont{Polino}},
  \bibinfo{author}{\bibfnamefont{D.}~\bibnamefont{Poderini}},
  \bibinfo{author}{\bibfnamefont{G.}~\bibnamefont{Rodari}},
  \bibinfo{author}{\bibfnamefont{I.}~\bibnamefont{Agresti}},
  \bibinfo{author}{\bibfnamefont{A.}~\bibnamefont{Suprano}},
  \bibinfo{author}{\bibfnamefont{G.}~\bibnamefont{Carvacho}},
  \bibinfo{author}{\bibfnamefont{E.}~\bibnamefont{Wolfe}},
  \bibinfo{author}{\bibfnamefont{A.}~\bibnamefont{Canabarro}},
  \bibinfo{author}{\bibfnamefont{G.}~\bibnamefont{Moreno}},
  \bibinfo{author}{\bibfnamefont{G.}~\bibnamefont{Milani}},
  \bibnamefont{et~al.}, \bibinfo{journal}{arXiv preprint arXiv:2210.07263}
  (\bibinfo{year}{2022}).

\bibitem[{\citenamefont{Brask and Chaves}(2017)}]{brask2017bell}
\bibinfo{author}{\bibfnamefont{J.~B.} \bibnamefont{Brask}} \bibnamefont{and}
  \bibinfo{author}{\bibfnamefont{R.}~\bibnamefont{Chaves}},
  \bibinfo{journal}{Journal of Physics A: Mathematical and Theoretical}
  \textbf{\bibinfo{volume}{50}}, \bibinfo{pages}{094001}
  (\bibinfo{year}{2017}).

\bibitem[{\citenamefont{Chaves et~al.}(2018)\citenamefont{Chaves, Carvacho,
  Agresti, Di~Giulio, Aolita, Giacomini, and Sciarrino}}]{chaves2018quantum}
\bibinfo{author}{\bibfnamefont{R.}~\bibnamefont{Chaves}},
  \bibinfo{author}{\bibfnamefont{G.}~\bibnamefont{Carvacho}},
  \bibinfo{author}{\bibfnamefont{I.}~\bibnamefont{Agresti}},
  \bibinfo{author}{\bibfnamefont{V.}~\bibnamefont{Di~Giulio}},
  \bibinfo{author}{\bibfnamefont{L.}~\bibnamefont{Aolita}},
  \bibinfo{author}{\bibfnamefont{S.}~\bibnamefont{Giacomini}},
  \bibnamefont{and}
  \bibinfo{author}{\bibfnamefont{F.}~\bibnamefont{Sciarrino}},
  \bibinfo{journal}{Nature Physics} \textbf{\bibinfo{volume}{14}},
  \bibinfo{pages}{291} (\bibinfo{year}{2018}).

\bibitem[{\citenamefont{Gachechiladze et~al.}(2020)\citenamefont{Gachechiladze,
  Miklin, and Chaves}}]{gachechiladze2020quantifying}
\bibinfo{author}{\bibfnamefont{M.}~\bibnamefont{Gachechiladze}},
  \bibinfo{author}{\bibfnamefont{N.}~\bibnamefont{Miklin}}, \bibnamefont{and}
  \bibinfo{author}{\bibfnamefont{R.}~\bibnamefont{Chaves}},
  \bibinfo{journal}{Physical Review Letters} \textbf{\bibinfo{volume}{125}},
  \bibinfo{pages}{230401} (\bibinfo{year}{2020}).

\bibitem[{\citenamefont{Agresti et~al.}(2022)\citenamefont{Agresti, Poderini,
  Polacchi, Miklin, Gachechiladze, Suprano, Polino, Milani, Carvacho, Chaves
  et~al.}}]{agresti2022experimental}
\bibinfo{author}{\bibfnamefont{I.}~\bibnamefont{Agresti}},
  \bibinfo{author}{\bibfnamefont{D.}~\bibnamefont{Poderini}},
  \bibinfo{author}{\bibfnamefont{B.}~\bibnamefont{Polacchi}},
  \bibinfo{author}{\bibfnamefont{N.}~\bibnamefont{Miklin}},
  \bibinfo{author}{\bibfnamefont{M.}~\bibnamefont{Gachechiladze}},
  \bibinfo{author}{\bibfnamefont{A.}~\bibnamefont{Suprano}},
  \bibinfo{author}{\bibfnamefont{E.}~\bibnamefont{Polino}},
  \bibinfo{author}{\bibfnamefont{G.}~\bibnamefont{Milani}},
  \bibinfo{author}{\bibfnamefont{G.}~\bibnamefont{Carvacho}},
  \bibinfo{author}{\bibfnamefont{R.}~\bibnamefont{Chaves}},
  \bibnamefont{et~al.}, \bibinfo{journal}{Science advances}
  \textbf{\bibinfo{volume}{8}}, \bibinfo{pages}{eabm1515}
  (\bibinfo{year}{2022}).

\bibitem[{\citenamefont{Evans}(2016{\natexlab{a}})}]{evans2016graphs}
\bibinfo{author}{\bibfnamefont{R.~J.} \bibnamefont{Evans}},
  \bibinfo{journal}{Scandinavian Journal of Statistics}
  \textbf{\bibinfo{volume}{43}}, \bibinfo{pages}{625}
  (\bibinfo{year}{2016}{\natexlab{a}}).

\bibitem[{\citenamefont{Ansanelli}(2022)}]{ansanelli2022observational}
\bibinfo{author}{\bibfnamefont{M.~M.} \bibnamefont{Ansanelli}}
  (\bibinfo{year}{2022}).

\bibitem[{\citenamefont{Garcia et~al.}(2005)\citenamefont{Garcia, Stillman, and
  Sturmfels}}]{garcia2005algebraic}
\bibinfo{author}{\bibfnamefont{L.~D.} \bibnamefont{Garcia}},
  \bibinfo{author}{\bibfnamefont{M.}~\bibnamefont{Stillman}}, \bibnamefont{and}
  \bibinfo{author}{\bibfnamefont{B.}~\bibnamefont{Sturmfels}},
  \bibinfo{journal}{Journal of Symbolic Computation}
  \textbf{\bibinfo{volume}{39}}, \bibinfo{pages}{331} (\bibinfo{year}{2005}).

\bibitem[{\citenamefont{Geiger and Meek}(2013)}]{geiger2013quantifier}
\bibinfo{author}{\bibfnamefont{D.}~\bibnamefont{Geiger}} \bibnamefont{and}
  \bibinfo{author}{\bibfnamefont{C.}~\bibnamefont{Meek}},
  \bibinfo{journal}{arXiv preprint arXiv:1301.6698}  (\bibinfo{year}{2013}).

\bibitem[{\citenamefont{Chaves et~al.}(2014)\citenamefont{Chaves, Luft, Maciel,
  Gross, Janzing, and Sch{\"o}lkopf}}]{chaves2014inferring}
\bibinfo{author}{\bibfnamefont{R.}~\bibnamefont{Chaves}},
  \bibinfo{author}{\bibfnamefont{L.}~\bibnamefont{Luft}},
  \bibinfo{author}{\bibfnamefont{T.}~\bibnamefont{Maciel}},
  \bibinfo{author}{\bibfnamefont{D.}~\bibnamefont{Gross}},
  \bibinfo{author}{\bibfnamefont{D.}~\bibnamefont{Janzing}}, \bibnamefont{and}
  \bibinfo{author}{\bibfnamefont{B.}~\bibnamefont{Sch{\"o}lkopf}}, in
  \emph{\bibinfo{booktitle}{Proceedings of the Thirtieth Conference on
  Uncertainty in Artificial Intelligence}} (\bibinfo{year}{2014}), pp.
  \bibinfo{pages}{112--121}.

\bibitem[{\citenamefont{Chaves}(2016)}]{chaves2016polynomial}
\bibinfo{author}{\bibfnamefont{R.}~\bibnamefont{Chaves}},
  \bibinfo{journal}{Physical review letters} \textbf{\bibinfo{volume}{116}},
  \bibinfo{pages}{010402} (\bibinfo{year}{2016}).

\bibitem[{\citenamefont{Kela et~al.}(2019)\citenamefont{Kela, Von~Prillwitz,
  {\AA}berg, Chaves, and Gross}}]{kela2019semidefinite}
\bibinfo{author}{\bibfnamefont{A.}~\bibnamefont{Kela}},
  \bibinfo{author}{\bibfnamefont{K.}~\bibnamefont{Von~Prillwitz}},
  \bibinfo{author}{\bibfnamefont{J.}~\bibnamefont{{\AA}berg}},
  \bibinfo{author}{\bibfnamefont{R.}~\bibnamefont{Chaves}}, \bibnamefont{and}
  \bibinfo{author}{\bibfnamefont{D.}~\bibnamefont{Gross}},
  \bibinfo{journal}{IEEE Transactions on Information Theory}
  \textbf{\bibinfo{volume}{66}}, \bibinfo{pages}{339} (\bibinfo{year}{2019}).

\bibitem[{\citenamefont{Lee and Spekkens}(2017)}]{lee2017causal}
\bibinfo{author}{\bibfnamefont{C.~M.} \bibnamefont{Lee}} \bibnamefont{and}
  \bibinfo{author}{\bibfnamefont{R.~W.} \bibnamefont{Spekkens}},
  \bibinfo{journal}{Journal of Causal Inference} \textbf{\bibinfo{volume}{5}}
  (\bibinfo{year}{2017}).

\bibitem[{\citenamefont{Wolfe et~al.}(2019{\natexlab{a}})\citenamefont{Wolfe,
  Spekkens, and Fritz}}]{WolfeSpekkensFritz+2019}
\bibinfo{author}{\bibfnamefont{E.}~\bibnamefont{Wolfe}},
  \bibinfo{author}{\bibfnamefont{R.~W.} \bibnamefont{Spekkens}},
  \bibnamefont{and} \bibinfo{author}{\bibfnamefont{T.}~\bibnamefont{Fritz}},
  \bibinfo{journal}{Journal of Causal Inference} \textbf{\bibinfo{volume}{7}},
  \bibinfo{pages}{20170020} (\bibinfo{year}{2019}{\natexlab{a}}),
  \urlprefix\url{https://doi.org/10.1515/jci-2017-0020}.

\bibitem[{\citenamefont{Pozas-Kerstjens
  et~al.}(2019)\citenamefont{Pozas-Kerstjens, Rabelo, Rudnicki, Chaves,
  Cavalcanti, Navascu{\'e}s, and Ac{\'\i}n}}]{pozas2019bounding}
\bibinfo{author}{\bibfnamefont{A.}~\bibnamefont{Pozas-Kerstjens}},
  \bibinfo{author}{\bibfnamefont{R.}~\bibnamefont{Rabelo}},
  \bibinfo{author}{\bibfnamefont{{\L}.}~\bibnamefont{Rudnicki}},
  \bibinfo{author}{\bibfnamefont{R.}~\bibnamefont{Chaves}},
  \bibinfo{author}{\bibfnamefont{D.}~\bibnamefont{Cavalcanti}},
  \bibinfo{author}{\bibfnamefont{M.}~\bibnamefont{Navascu{\'e}s}},
  \bibnamefont{and}
  \bibinfo{author}{\bibfnamefont{A.}~\bibnamefont{Ac{\'\i}n}},
  \bibinfo{journal}{Physical review letters} \textbf{\bibinfo{volume}{123}},
  \bibinfo{pages}{140503} (\bibinfo{year}{2019}).

\bibitem[{\citenamefont{Kriv{\'a}chy et~al.}(2020)\citenamefont{Kriv{\'a}chy,
  Cai, Cavalcanti, Tavakoli, Gisin, and Brunner}}]{krivachy2020neural}
\bibinfo{author}{\bibfnamefont{T.}~\bibnamefont{Kriv{\'a}chy}},
  \bibinfo{author}{\bibfnamefont{Y.}~\bibnamefont{Cai}},
  \bibinfo{author}{\bibfnamefont{D.}~\bibnamefont{Cavalcanti}},
  \bibinfo{author}{\bibfnamefont{A.}~\bibnamefont{Tavakoli}},
  \bibinfo{author}{\bibfnamefont{N.}~\bibnamefont{Gisin}}, \bibnamefont{and}
  \bibinfo{author}{\bibfnamefont{N.}~\bibnamefont{Brunner}},
  \bibinfo{journal}{npj Quantum Information} \textbf{\bibinfo{volume}{6}},
  \bibinfo{pages}{1} (\bibinfo{year}{2020}).

\bibitem[{\citenamefont{Pearl}(1995)}]{pearl1995testability}
\bibinfo{author}{\bibfnamefont{J.}~\bibnamefont{Pearl}}, in
  \emph{\bibinfo{booktitle}{Proceedings of the Eleventh conference on
  Uncertainty in artificial intelligence}} (\bibinfo{year}{1995}), pp.
  \bibinfo{pages}{435--443}.

\bibitem[{\citenamefont{Popescu and Rohrlich}(1994)}]{popescu1994quantum}
\bibinfo{author}{\bibfnamefont{S.}~\bibnamefont{Popescu}} \bibnamefont{and}
  \bibinfo{author}{\bibfnamefont{D.}~\bibnamefont{Rohrlich}},
  \bibinfo{journal}{Foundations of Physics} \textbf{\bibinfo{volume}{24}},
  \bibinfo{pages}{379} (\bibinfo{year}{1994}).

\bibitem[{\citenamefont{Wolfe et~al.}(2019{\natexlab{b}})\citenamefont{Wolfe,
  Spekkens, and Fritz}}]{wolfe2019inflation}
\bibinfo{author}{\bibfnamefont{E.}~\bibnamefont{Wolfe}},
  \bibinfo{author}{\bibfnamefont{R.~W.} \bibnamefont{Spekkens}},
  \bibnamefont{and} \bibinfo{author}{\bibfnamefont{T.}~\bibnamefont{Fritz}},
  \bibinfo{journal}{Journal of Causal Inference} \textbf{\bibinfo{volume}{7}}
  (\bibinfo{year}{2019}{\natexlab{b}}).

\bibitem[{\citenamefont{Evans}(2012)}]{Evans2012}
\bibinfo{author}{\bibfnamefont{R.~J.} \bibnamefont{Evans}}, in
  \emph{\bibinfo{booktitle}{2012 IEEE International Workshop on Machine
  Learning for Signal Processing}} (\bibinfo{year}{2012}), pp.
  \bibinfo{pages}{1--6}.

\bibitem[{\citenamefont{T and J}()}]{verma1988causal}
\bibinfo{author}{\bibfnamefont{V.}~\bibnamefont{T}} \bibnamefont{and}
  \bibinfo{author}{\bibfnamefont{P.}~\bibnamefont{J}}, in
  \emph{\bibinfo{booktitle}{Proc. 4th Workshop on Uncertainty in Artificial
  Intelligence (Minneapolis, MN and Mountain View, CA)}} (????).

\bibitem[{\citenamefont{Chiribella et~al.}(2010)\citenamefont{Chiribella,
  D'Ariano, and Perinotti}}]{PhysRevA.81.062348}
\bibinfo{author}{\bibfnamefont{G.}~\bibnamefont{Chiribella}},
  \bibinfo{author}{\bibfnamefont{G.~M.} \bibnamefont{D'Ariano}},
  \bibnamefont{and}
  \bibinfo{author}{\bibfnamefont{P.}~\bibnamefont{Perinotti}},
  \bibinfo{journal}{Phys. Rev. A} \textbf{\bibinfo{volume}{81}},
  \bibinfo{pages}{062348} (\bibinfo{year}{2010}),
  \urlprefix\url{https://link.aps.org/doi/10.1103/PhysRevA.81.062348}.

\bibitem[{\citenamefont{Evans}(2016{\natexlab{b}})}]{Evans2016}
\bibinfo{author}{\bibfnamefont{R.~J.} \bibnamefont{Evans}},
  \bibinfo{journal}{Scandinavian Journal of Statistics}
  \textbf{\bibinfo{volume}{43}}, \bibinfo{pages}{625}
  (\bibinfo{year}{2016}{\natexlab{b}}),
  \eprint{https://onlinelibrary.wiley.com/doi/pdf/10.1111/sjos.12194},
  \urlprefix\url{https://onlinelibrary.wiley.com/doi/abs/10.1111/sjos.12194}.

\bibitem[{\citenamefont{Branciard
  et~al.}(2010{\natexlab{b}})\citenamefont{Branciard, Gisin, and
  Pironio}}]{PhysRevLett.104.170401}
\bibinfo{author}{\bibfnamefont{C.}~\bibnamefont{Branciard}},
  \bibinfo{author}{\bibfnamefont{N.}~\bibnamefont{Gisin}}, \bibnamefont{and}
  \bibinfo{author}{\bibfnamefont{S.}~\bibnamefont{Pironio}},
  \bibinfo{journal}{Phys. Rev. Lett.} \textbf{\bibinfo{volume}{104}},
  \bibinfo{pages}{170401} (\bibinfo{year}{2010}{\natexlab{b}}),
  \urlprefix\url{https://link.aps.org/doi/10.1103/PhysRevLett.104.170401}.

\bibitem[{\citenamefont{Branciard et~al.}(2012)\citenamefont{Branciard, Rosset,
  Gisin, and Pironio}}]{PhysRevA.85.032119}
\bibinfo{author}{\bibfnamefont{C.}~\bibnamefont{Branciard}},
  \bibinfo{author}{\bibfnamefont{D.}~\bibnamefont{Rosset}},
  \bibinfo{author}{\bibfnamefont{N.}~\bibnamefont{Gisin}}, \bibnamefont{and}
  \bibinfo{author}{\bibfnamefont{S.}~\bibnamefont{Pironio}},
  \bibinfo{journal}{Phys. Rev. A} \textbf{\bibinfo{volume}{85}},
  \bibinfo{pages}{032119} (\bibinfo{year}{2012}),
  \urlprefix\url{https://link.aps.org/doi/10.1103/PhysRevA.85.032119}.

\bibitem[{\citenamefont{Zukowski et~al.}(1993)\citenamefont{Zukowski,
  Zeilinger, Horne, and Ekert}}]{zukowski1993event}
\bibinfo{author}{\bibfnamefont{M.}~\bibnamefont{Zukowski}},
  \bibinfo{author}{\bibfnamefont{A.}~\bibnamefont{Zeilinger}},
  \bibinfo{author}{\bibfnamefont{M.~A.} \bibnamefont{Horne}}, \bibnamefont{and}
  \bibinfo{author}{\bibfnamefont{A.~K.} \bibnamefont{Ekert}},
  \bibinfo{journal}{Physical Review Letters} \textbf{\bibinfo{volume}{71}}
  (\bibinfo{year}{1993}).

\bibitem[{\citenamefont{Bonet}(2013)}]{bonet2013instrumentality}
\bibinfo{author}{\bibfnamefont{B.}~\bibnamefont{Bonet}},
  \bibinfo{journal}{arXiv preprint arXiv:1301.2258}  (\bibinfo{year}{2013}).

\bibitem[{\citenamefont{Rosset et~al.}(2018)\citenamefont{Rosset, Gisin, and
  Wolfe}}]{Rosset2018UniversalBO}
\bibinfo{author}{\bibfnamefont{D.}~\bibnamefont{Rosset}},
  \bibinfo{author}{\bibfnamefont{N.}~\bibnamefont{Gisin}}, \bibnamefont{and}
  \bibinfo{author}{\bibfnamefont{E.}~\bibnamefont{Wolfe}},
  \bibinfo{journal}{Quantum Inf. Comput.} \textbf{\bibinfo{volume}{18}},
  \bibinfo{pages}{910} (\bibinfo{year}{2018}).

\bibitem[{\citenamefont{Lasserre}(2001)}]{Lasserre2001GlobalOW}
\bibinfo{author}{\bibfnamefont{J.~B.} \bibnamefont{Lasserre}},
  \bibinfo{journal}{SIAM J. Optim.} \textbf{\bibinfo{volume}{11}},
  \bibinfo{pages}{796} (\bibinfo{year}{2001}).

\bibitem[{\citenamefont{Navascués and
  Wolfe}(2020)}]{NavascuesWolfe+2020+70+91}
\bibinfo{author}{\bibfnamefont{M.}~\bibnamefont{Navascués}} \bibnamefont{and}
  \bibinfo{author}{\bibfnamefont{E.}~\bibnamefont{Wolfe}},
  \bibinfo{journal}{Journal of Causal Inference} \textbf{\bibinfo{volume}{8}},
  \bibinfo{pages}{70} (\bibinfo{year}{2020}),
  \urlprefix\url{https://doi.org/10.1515/jci-2018-0008}.

\bibitem[{\citenamefont{Fine}(1982)}]{fine1982hidden}
\bibinfo{author}{\bibfnamefont{A.}~\bibnamefont{Fine}},
  \bibinfo{journal}{Physical Review Letters} \textbf{\bibinfo{volume}{48}},
  \bibinfo{pages}{291} (\bibinfo{year}{1982}).

\bibitem[{\citenamefont{{Gurobi Optimization, LLC}}(2022)}]{gurobi}
\bibinfo{author}{\bibnamefont{{Gurobi Optimization, LLC}}},
  \emph{\bibinfo{title}{{Gurobi Optimizer Reference Manual}}}
  (\bibinfo{year}{2022}), \urlprefix\url{https://www.gurobi.com}.

\bibitem[{\citenamefont{Miklin et~al.}(2022)\citenamefont{Miklin,
  Gachechiladze, Moreno, and Chaves}}]{miklin2022causal}
\bibinfo{author}{\bibfnamefont{N.}~\bibnamefont{Miklin}},
  \bibinfo{author}{\bibfnamefont{M.}~\bibnamefont{Gachechiladze}},
  \bibinfo{author}{\bibfnamefont{G.}~\bibnamefont{Moreno}}, \bibnamefont{and}
  \bibinfo{author}{\bibfnamefont{R.}~\bibnamefont{Chaves}},
  \bibinfo{journal}{Journal of Causal Inference} \textbf{\bibinfo{volume}{10}},
  \bibinfo{pages}{45} (\bibinfo{year}{2022}).

\bibitem[{\citenamefont{Boyd and Vandenberghe}(2006)}]{BoydVandenberghe}
\bibinfo{author}{\bibfnamefont{S.}~\bibnamefont{Boyd}} \bibnamefont{and}
  \bibinfo{author}{\bibfnamefont{L.}~\bibnamefont{Vandenberghe}}
  (\bibinfo{year}{2006}).

\bibitem[{\citenamefont{Beigi and Renou}(2021)}]{beigi2021covariance}
\bibinfo{author}{\bibfnamefont{S.}~\bibnamefont{Beigi}} \bibnamefont{and}
  \bibinfo{author}{\bibfnamefont{M.-O.} \bibnamefont{Renou}},
  \bibinfo{journal}{IEEE Transactions on Information Theory}
  \textbf{\bibinfo{volume}{68}}, \bibinfo{pages}{384} (\bibinfo{year}{2021}).

\bibitem[{\citenamefont{Gisin et~al.}(2020)\citenamefont{Gisin, Bancal, Cai,
  Remy, Tavakoli, Zambrini~Cruzeiro, Popescu, and Brunner}}]{GisinBancalCai}
\bibinfo{author}{\bibfnamefont{N.}~\bibnamefont{Gisin}},
  \bibinfo{author}{\bibfnamefont{J.-D.} \bibnamefont{Bancal}},
  \bibinfo{author}{\bibfnamefont{Y.}~\bibnamefont{Cai}},
  \bibinfo{author}{\bibfnamefont{P.}~\bibnamefont{Remy}},
  \bibinfo{author}{\bibfnamefont{A.}~\bibnamefont{Tavakoli}},
  \bibinfo{author}{\bibfnamefont{E.}~\bibnamefont{Zambrini~Cruzeiro}},
  \bibinfo{author}{\bibfnamefont{S.}~\bibnamefont{Popescu}}, \bibnamefont{and}
  \bibinfo{author}{\bibfnamefont{N.}~\bibnamefont{Brunner}},
  \bibinfo{journal}{Nature Communications} \textbf{\bibinfo{volume}{11}},
  \bibinfo{pages}{2378} (\bibinfo{year}{2020}),
  \urlprefix\url{https://doi.org/10.1038/s41467-020-16137-4}.

\bibitem[{\citenamefont{Renou et~al.}(2019{\natexlab{b}})\citenamefont{Renou,
  Wang, Boreiri, Beigi, Gisin, and Brunner}}]{renou2019limits}
\bibinfo{author}{\bibfnamefont{M.-O.} \bibnamefont{Renou}},
  \bibinfo{author}{\bibfnamefont{Y.}~\bibnamefont{Wang}},
  \bibinfo{author}{\bibfnamefont{S.}~\bibnamefont{Boreiri}},
  \bibinfo{author}{\bibfnamefont{S.}~\bibnamefont{Beigi}},
  \bibinfo{author}{\bibfnamefont{N.}~\bibnamefont{Gisin}}, \bibnamefont{and}
  \bibinfo{author}{\bibfnamefont{N.}~\bibnamefont{Brunner}},
  \bibinfo{journal}{Physical review letters} \textbf{\bibinfo{volume}{123}},
  \bibinfo{pages}{070403} (\bibinfo{year}{2019}{\natexlab{b}}).

\bibitem[{\citenamefont{Finkelstein et~al.}(2021)\citenamefont{Finkelstein,
  Zjawin, Wolfe, Shpitser, and Spekkens}}]{Finkelstein2021EntropicIC}
\bibinfo{author}{\bibfnamefont{N.}~\bibnamefont{Finkelstein}},
  \bibinfo{author}{\bibfnamefont{B.}~\bibnamefont{Zjawin}},
  \bibinfo{author}{\bibfnamefont{E.}~\bibnamefont{Wolfe}},
  \bibinfo{author}{\bibfnamefont{I.}~\bibnamefont{Shpitser}}, \bibnamefont{and}
  \bibinfo{author}{\bibfnamefont{R.~W.} \bibnamefont{Spekkens}}, in
  \emph{\bibinfo{booktitle}{UAI}} (\bibinfo{year}{2021}).

\bibitem[{\citenamefont{Balke and Pearl}(1997)}]{balke1997bounds}
\bibinfo{author}{\bibfnamefont{A.}~\bibnamefont{Balke}} \bibnamefont{and}
  \bibinfo{author}{\bibfnamefont{J.}~\bibnamefont{Pearl}},
  \bibinfo{journal}{Journal of the American Statistical Association}
  \textbf{\bibinfo{volume}{92}}, \bibinfo{pages}{1171} (\bibinfo{year}{1997}).

\bibitem[{\citenamefont{Bixby}(2007)}]{bixby2007gurobi}
\bibinfo{author}{\bibfnamefont{B.}~\bibnamefont{Bixby}},
  \bibinfo{journal}{Transp. Re-search Part B} \textbf{\bibinfo{volume}{41}},
  \bibinfo{pages}{159} (\bibinfo{year}{2007}).

\end{thebibliography}
\appendix
\onecolumngrid

\section{Classical simulation of the PR-box with symmetric instrument}
\label{app: Classical_PR-Box}

We show how the correlations of a PR-Box \ref{eq:NS_Box} between Bob and Charlie, with uniform distribution of Alice can be simulated classically. The distribution under scrutiny is 
\begin{equation}
    p(a,b,c)=p(a)p(bc|a)=\begin{cases}
    1/6 &\text{ if } c=f(a,b)+b\\
    \quad 0 &\text{ otherwise }.
\end{cases}
\end{equation}
It is enough to take $\lambda=\lambda_0\lambda_1$, where $\lambda_i\in \{0,1,2\}$ and $\mu=\mu_0\mu_1$ where $\mu_i\in \{0,1\}$, with distributions 
\begin{equation}
    p(\gamma)=([01]+[11]+[20]+[22])/4 \quad \text{ and }\quad  p(\alpha)=[10]/3+2[01]/3,
\end{equation}
where $[a_0a_1]=\delta_{\lambda_0,a_0}\delta_{\lambda_1,a_1}$ and $[c_0c_1]=\delta_{\mu_0,c_0}\delta_{\mu_1,c_1}$. Define deterministic response functions for Alice and Charlie 
\begin{eqnarray}
\nonumber
    p(a|\lambda,b)=\delta_{a,\lambda_b}\\
p(c|\mu,b)=\delta_{c,\mu_b}.
\end{eqnarray}
Bob also has a deterministic local response function 
\begin{equation}
p(b|\lambda\mu)=\begin{cases}
    1 \quad \text{ if } \mu_b=f(\lambda_b,b)+b\\
     0\quad  \text{ otherwise}
\end{cases}   .
\end{equation}
Because of our choice of strategies, this response function is well defined, i.e., the expression $\mu_b=f(\lambda_b,b)+b$ has a unique solution $b$, for every strategy $\lambda\mu$ given to Bob. For our case, since $b$ and $\mu_i$ are bits, we can write $b=\mu_b+f(\lambda_b,b)$, by summing $b+\mu_b$ in both sides. The deterministic response can be rewritten as 
\begin{equation}
    p(b|\lambda\mu)=\delta_{b,\mu_b+f(\lambda_b,b)}.
\end{equation}

Summing up, this implies that
\begin{equation}
    p(a,b,c)=\sum_{\lambda,\mu}p(\lambda)p(\mu) \delta_{a,\lambda_b} \delta_{b,\mu_b+f(\lambda_b,b)}\delta_{c,\mu_b}=\delta_{b,c+f(a,b)}\left(\sum_{\lambda,\mu}p(\lambda_b=a)p(\mu_b=c)\right),
\end{equation}
and we can check that for each case where $b=c+f(a,b)$ the sum above is 1/6, and thus this classical model recovers the distribution $p(a,b,c)$.

\section{Inequality derived from inflation with e-separation }
\label{app: IneqEvansInflation}
Here we explicitly  show the inequality derived using 3rd-order inflation after the node splitting operation in the Evans scenario. The inequality is obtained via the dual solution of the linear program detailed in section \ref{sec: e-sep+inflation technique}. We consider the inflation shown in Fig. \ref{fig: Inflation3rd}, where the compatibility test can be cast as a linear program on the joint probability distribution
\begin{equation}
    q(a_1,a_2,a_3,b_{11},b_{12},b_{13},b_{21},b_{22},b_{23},b_{31},b_{32},b_{33},c_1,c_2,c_3\vert b^{\#}_1,b^{\#}_2,b^{\#}_3),
\end{equation}
that must respect non-negativity and normalization. Defining
\begin{equation}
\begin{aligned}
  &q_3(a_1,a_2,a_3,b_{11},b_{22},b_{33},c_1,c_2,c_3|b^{\#}_1,b^{\#}_2,b^{\#}_3)=\sum_{b_{i\neq j}}q(a_1,...,b_{11},...c_3\vert b^{\#}_1,b^{\#}_2,b^{\#}_3)
  \end{aligned},
\end{equation}
the factorization constraints becomes
\begin{equation}
\begin{aligned}
q_3(a_1,a_2,a_3,b_{11},b_{22},b_{33},c_1,c_2,c_3\vert (b^{\#}_1,b^{\#}_2,b^{\#}_3)=(b_{11},b_{22},b_{33}))=\prod_{i=1,2,3}p(a_i,b_{ii},c_i).
\end{aligned}
\end{equation}
Where $p(a,b,c)$ is the probability distribution under scrutiny. Ultimately, we need to impose symmetry constraints implied by the global markov model of the inflated network, these conditions will be simply equality constraints on the entries of $q$. This can be done in analogy to \ref{eq: symmetryinflation} and \ref{eq: symmetryinflation1} for when $(b^{\#}_1,b^{\#}_2=b^{\#}_3)$,  $(b^{\#}_1=b^{\#}_2,b^{\#}_3)$ and $(b^{\#}_1=b^{\#}_2=b^{\#}_3).$
Due to the factorization
conditions imposed in the diagonal distribution $q_3$, the obtained inequality contains the corresponding cubic elements and is given by
\begin{equation}
\label{eq: witness-inflation}
    \begin{split}
        &p_{ABC}(0,0,0)^3/3 + p_{ABC}(0,0,0)^2p_{ABC}(0,1,1) + p_{ABC}(0,0,0)^2p_{ABC}(1,0,0) + p_{ABC}(0,0,0)^2p_{ABC}(1,1,0)  \\
        &+ p_{ABC}(0,0,0)^2p_{ABC}(2,0,1) + p_{ABC}(0,0,0)^2p_{ABC}(2,1,1) + p_{ABC}(0,0,0)p_{ABC}(0,1,1)^2 + p_{ABC}(0,0,0)p_{ABC}(2,1,1)^2 \\
        &+2p_{ABC}(0,0,0)p_{ABC}(0,1,1)p_{ABC}(1,0,0) + p_{ABC}(0,1,1)^3/3 + p_{ABC}(0,1,1)^2p_{ABC}(1,0,0) + p_{ABC}(0,1,1)^2p_{ABC}(1,1,0)\\
        &+ 2p_{ABC}(0,0,0)p_{ABC}(0,1,1)p_{ABC}(2,0,1)+ 2p_{ABC}(0,0,0)p_{ABC}(0,1,1)p_{ABC}(1,1,0)+ p_{ABC}(2,0,1)^2p_{ABC}(2,1,1)\\
        &+ 2p_{ABC}(0,0,0)p_{ABC}(1,0,0)p_{ABC}(1,1,0)  + 2p_{ABC}(0,0,0)p_{ABC}(1,0,0)p_{ABC}(2,0,1) + p_{ABC}(2,0,1)p_{ABC}(2,1,1)^2 \\
        & + 2p_{ABC}(0,0,0)p_{ABC}(1,0,0)p_{ABC}(2,1,1) + p_{ABC}(0,0,0)p_{ABC}(1,1,0)^2  + p_{ABC}(0,0,0)p_{ABC}(1,0,0)^2 \\
        &+ 2p_{ABC}(0,0,0)p_{ABC}(1,1,0)p_{ABC}(2,0,1)  + 2p_{ABC}(0,0,0)p_{ABC}(1,1,0)p_{ABC}(2,1,1) - 5p_{ABC}(0,0,0)p_{ABC}(2,0,1)^2/3 \\ 
        &  + p_{ABC}(0,1,1)^2p_{ABC}(2,0,1) + p_{ABC}(0,1,1)^2p_{ABC}(2,1,1) + 2p_{ABC}(0,0,0)p_{ABC}(2,0,1)p_{ABC}(2,1,1) \\
        &+ p_{ABC}(0,1,1)p_{ABC}(1,0,0)^2 + 2p_{ABC}(0,1,1)p_{ABC}(1,0,0)p_{ABC}(1,1,0) + p_{ABC}(1,1,0)p_{ABC}(2,1,1)^2 \\
        &+ 2p_{ABC}(0,1,1)p_{ABC}(1,0,0)p_{ABC}(2,0,1)  + 2p_{ABC}(0,1,1)p_{ABC}(1,0,0)p_{ABC}(2,1,1) - p_{ABC}(2,0,1)^3 \\
        &+ p_{ABC}(0,1,1)p_{ABC}(1,1,0)^2 + 14p_{ABC}(0,1,1)p_{ABC}(1,1,0)p_{ABC}(2,0,1)/3 + 2p_{ABC}(0,0,0)p_{ABC}(0,1,1)p_{ABC}(2,1,1)\\
        &+ 2p_{ABC}(0,1,1)p_{ABC}(1,1,0)p_{ABC}(2,1,1) + p_{ABC}(0,1,1)p_{ABC}(2,0,1)^2 + 2p_{ABC}(0,1,1)p_{ABC}(2,0,1)p_{ABC}(2,1,1) \\
        & + p_{ABC}(0,1,1)p_{ABC}(2,1,1)^2 + p_{ABC}(1,0,0)^3/3 + p_{ABC}(1,0,0)p_{ABC}(2,1,1)^2 + p_{ABC}(1,1,0)^3/3 \\
        & + p_{ABC}(1,0,0)^2p_{ABC}(1,1,0) + p_{ABC}(1,0,0)^2p_{ABC}(2,0,1) + p_{ABC}(1,0,0)^2p_{ABC}(2,1,1)+ p_{ABC}(2,1,1)^3/3 \\
        &+ p_{ABC}(1,0,0)p_{ABC}(1,1,0)^2 + 2p_{ABC}(1,0,0)p_{ABC}(1,1,0)p(2,0,1)+ 14p_{ABC}(1,1,0)p_{ABC}(2,0,1)p_{ABC}(2,1,1)/3\\
        & + 2p_{ABC}(1,0,0)p_{ABC}(1,1,0)p_{ABC}(2,1,1) - 5p_{ABC}(1,0,0)p_{ABC}(2,0,1)^2/3  + 2p_{ABC}(1,0,0)p_{ABC}(2,0,1)p_{ABC}(2,1,1)\\
        & + p_{ABC}(1,1,0)^2p_{ABC}(2,0,1) + p_{ABC}(1,1,0)^2p_{ABC}(2,1,1) + p_{ABC}(1,1,0)p_{ABC}(2,0,1)^2  \leq 1/3.
    \end{split}
\end{equation}
For simplicity we have assumed $p(c=f(a,b)+b)=1$, to keep only the nonzero terms of the inequality for our candidate distribution. This condition is also not too restrictive as it can be guaranteed under classical models, for example $p(0,0,0)=1$ trivially satisfies this constraint. If we choose $p_A(0)=p_A(2)=10/21$ and $p_A(2)=1/21$ the distribution violates the inequality \ref{eq: witness-inflation} by $\beta_{PR} \approx 3.2 \times 10^{-3}$ and has a very low resistance when combined with the uniform distribution with visibility $v\approx 0.9984$.
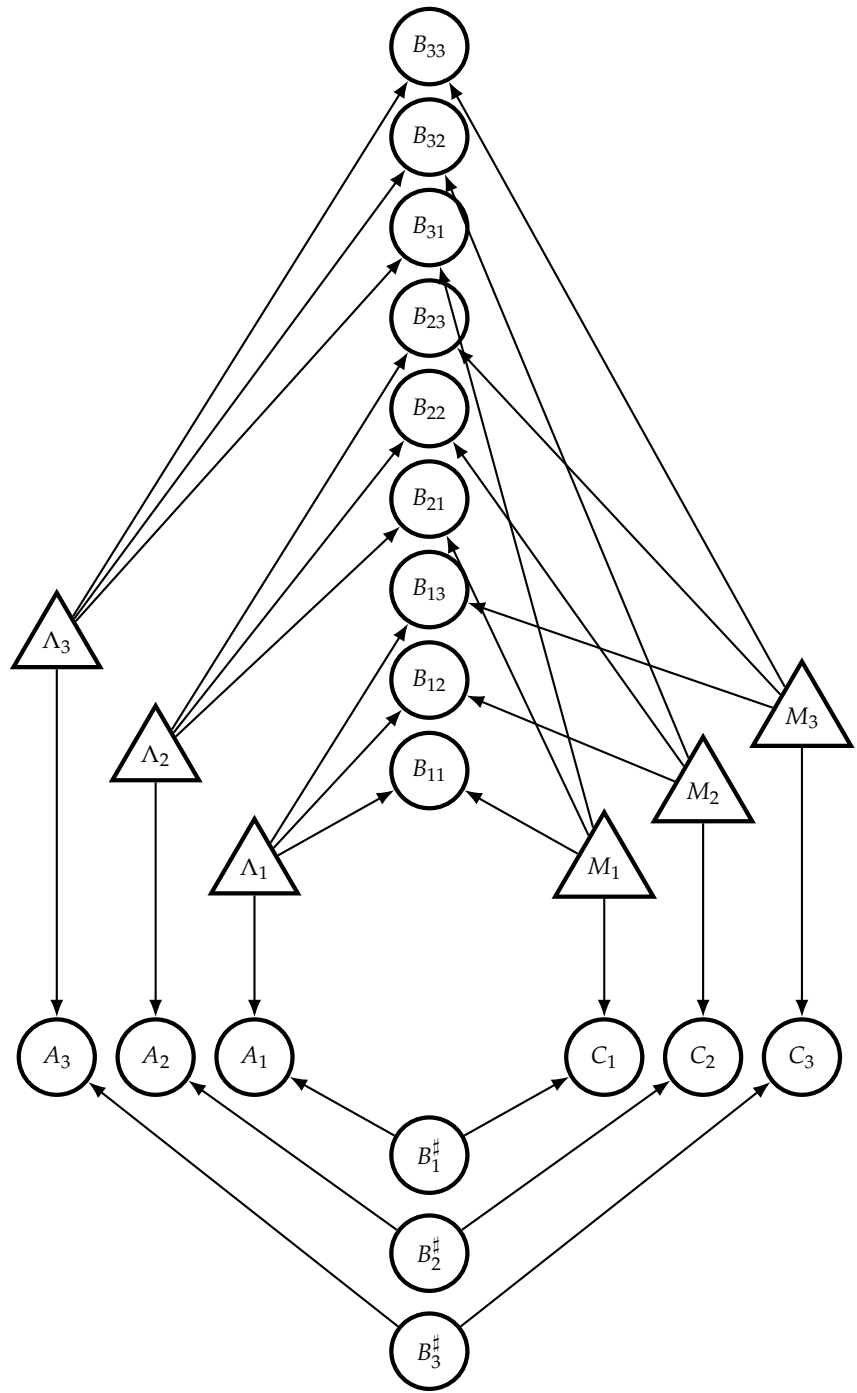
\begin{figure*}
\begin{tikzpicture}

	\foreach \x in {1,2,3} { 
		\foreach \y in {1,2,3} {
		    \node[var] (b\x\y) at (0, 1.2*3*\x+1.2*\y -1) {$B_{\x\y}$};
	}}
	\foreach \x in {1,2,3} {
	    \node[var] (a\x) at (-1-1.3*\x,0) {$A_\x$};
	    \node[latent] (l\x) at (-1-1.3*\x, 1+1.5*\x) {$\Lambda_\x$};
		\path[dir] (l\x) edge (a\x);
		\foreach \y in {1,2,3} \path[dir] (l\x) edge (b\x\y);
	}
	\foreach \y in {1,2,3} {
		\node[var] (c\y) at (1+1.3*\y,0) {$C_\y$};
	    \node[latent] (m\y) at (1+1.3*\y, 1.5+\y) {$M_\y$};
		\path[dir] (m\y) edge (c\y);
		\foreach \x in {1,2,3} \path[dir] (m\y) edge (b\x\y);
	}
	\foreach \z in {1,2,3} {
		    \node[var] (bs\z) at (0,-1.3*\z) {$B^\sharp_\z$};
			\path[dir] (bs\z) edge (a\z) (bs\z) edge (c\z); 
	}

\end{tikzpicture}
\caption{\textbf{Inflated Evans scenario} where $\Lambda_i$ and $M_i$ are copies of the original sources $\Lambda$ and $M$ and $A_i$ ,$B_{i,j}$, $C_i$ and $B^{\#}_i$ are copies of the original observable variables $A$ ,$B$, $C$ and $B^{\#}$. The indices of $A$, $B$ and $C$ indicate the latent dependency of each observable copy variable.}
\label{fig: Inflation3rd}
\end{figure*}

\end{document}